\pdfoutput=1
\PassOptionsToPackage{
pdfencoding=auto,
pdfnewwindow=true,
pdfusetitle=false,
bookmarks=true,
bookmarksnumbered=true,
bookmarksopen=true,
pdfpagemode=UseThumbs,
bookmarksopenlevel=1,
pdfpagelabels=false,
breaklinks=true,  
backref=false,
colorlinks=true,
}{hyperref}
\RequirePackage{etex}
\PassOptionsToPackage{skip=0pt, font=small,labelfont=bf}{caption}
\documentclass[10pt,showpacs,letterpaper,aps,pra,twocolumn,nofootinbib,superscriptaddress,floatfix]{revtex4-2}
\usepackage[utf8]{inputenc} 
\usepackage{hyperref}

\usepackage{amsthm,amsmath,amssymb,graphicx,comment,dsfont} 
\usepackage{color}
\usepackage{bbold}
\usepackage[none]{hyphenat}
\usepackage{fancyhdr}
\newtheorem{theorem}{Theorem}
\newtheorem{thmrepeat}{Theorem}
\newtheorem{prop}{Proposition}
\newtheorem{corollary}{Corollary}[theorem]
\newtheorem{lemma}[theorem]{Lemma}
\newtheorem{rem}[theorem]{Remark}
\newtheorem{definition}[theorem]{Definition}

\newcommand{\ket}[1]{\left| #1 \right\rangle}

\newcommand{\beq}{\begin{equation}}
\newcommand{\eeq}{\end{equation}}
\newcommand{\bea}{\begin{align}}
\newcommand{\eea}{\end{align}}

\usepackage{graphicx}
\usepackage[usenames,dvipsnames,table]{xcolor} 
\usepackage{mathtools}
\usepackage[normalem]{ulem}
\input{preamble}

\definecolor{googleblue}{RGB}{34, 0, 204}
\definecolor{panblue}{RGB}{0,24,150}
\definecolor{carmine}{RGB}{150, 0, 24}
\hypersetup{
unicode=true,
bookmarksnumbered=false,
bookmarksopen=false,
breaklinks=false,
colorlinks=true,
linkcolor=carmine,
citecolor=googleblue,
urlcolor=panblue,
anchorcolor=OliveGreen}

\newcommand{\blk}{\color{black}}

\newcommand{\projP}{%
\InputIfFileExists{Diagrams/projP.tikz}{}{\input{./figures/Diagrams/projP.tikz}}}
\newcommand{\projM}{%
\InputIfFileExists{Diagrams/projM.tikz}{}{\input{./figures/Diagrams/projM.tikz}}}
\newcommand{\incP}{%
\InputIfFileExists{Diagrams/incP.tikz}{}{\input{./figures/Diagrams/incP.tikz}}}
\newcommand{\incM}{%
\InputIfFileExists{Diagrams/incM.tikz}{}{\input{./figures/Diagrams/incM.tikz}}}

\allowdisplaybreaks

\usepackage{microtype}
\microtypecontext{spacing=nonfrench}
\microtypesetup{
expansion={true,nocompatibility},
protrusion={true,nocompatibility},
activate={true,nocompatibility},
tracking=true,
kerning=true,
spacing={true}
}

\usepackage[all=normal,floats=tight,mathspacing=tight,wordspacing=tight,paragraphs=normal,tracking=tight,charwidths=tight,mathdisplays=normal,sections=normal,margins=normal]{savetrees}

\everypar=\expandafter{\the\everypar\loosness=-1 }
\newcommand{\nocontentsline}[3]{}
\let\oldaddcontentsline\addcontentsline
\newcommand{\tocless}[2]{%
  \let\addcontentsline=\nocontentsline#1{#2}
  \let\addcontentsline\oldaddcontentsline}

\begin{document}
\title{Accessible fragments of generalized probabilistic theories, cone equivalence, and applications to witnessing nonclassicality}
\author{John H. Selby}
\email{john.h.selby@gmail.com}
\affiliation{International Centre for Theory of Quantum Technologies, University of Gda\'nsk, 80-309 Gda\'nsk, Poland}
\author{David Schmid}
\affiliation{International Centre for Theory of Quantum Technologies, University of Gda\'nsk, 80-309 Gda\'nsk, Poland}
\affiliation{Perimeter Institute for Theoretical Physics, 31 Caroline Street North, Waterloo, Ontario Canada N2L 2Y5}
\affiliation{Institute for Quantum Computing and Department of Physics and Astronomy, University of Waterloo, Waterloo, Ontario N2L 3G1, Canada}
\author{Elie Wolfe}
\affiliation{Perimeter Institute for Theoretical Physics, 31 Caroline Street North, Waterloo, Ontario Canada N2L 2Y5}
\author{Ana Bel\'en Sainz}
\affiliation{International Centre for Theory of Quantum Technologies, University of Gda\'nsk, 80-309 Gda\'nsk, Poland}
\author{Ravi Kunjwal}
\affiliation{Centre for Quantum Information and Communication, Ecole polytechnique de Bruxelles,
	CP 165, Universit\'e libre de Bruxelles, 1050 Brussels, Belgium}
\author{Robert W. Spekkens}
\affiliation{Perimeter Institute for Theoretical Physics, 31 Caroline Street North, Waterloo, Ontario Canada N2L 2Y5}

\date{\today}

\begin{abstract}
The formalism of generalized probabilistic theories
 (GPTs) was originally developed as a way to characterize the landscape of conceivable physical theories. Thus, the GPT describing a given physical theory necessarily includes {\em all} physically possible processes. We here consider the question of how to provide a GPT-like characterization of {\em a particular experimental setup} within a given physical theory. We show that the resulting characterization is {\em not} generally a GPT in and of itself---rather, it is described by a more general mathematical object that we introduce and term an {\em accessible GPT fragment}. We then introduce an equivalence relation, termed cone equivalence, between accessible GPT fragments 
(and, as a special case, between standard GPTs).
We give a number of examples of experimental scenarios that are best described using accessible GPT fragments, and where moreover cone-equivalence arises naturally. We then prove that an accessible GPT fragment admits of a classical explanation if and only if every other fragment that is cone-equivalent to it also admits of a classical explanation. Finally, we leverage this result to prove several fundamental results regarding the experimental requirements for witnessing the failure of generalized noncontextuality. In particular, we prove that neither incompatibility among measurements nor the assumption of freedom of choice is necessary for witnessing failures of generalized noncontextuality, and, moreover, that such failures can be witnessed even using arbitrarily inefficient detectors. 
\end{abstract}

\maketitle
\tableofcontents

\section{Introduction}

The framework of generalized probabilistic theories~\cite{Hardy,GPT_Barrett,chiribella2010probabilistic}, or GPTs, was developed in order to study conceivable theories of nature within a single landscape of possible theories.  The framework is built on a minimal set of operational principles, based primarily on the fact that any physical theory must, at a minimum, make probabilistic predictions about the outcomes of experiments. It should be noted that for single systems (including infinite-dimensional ones) the framework has a much longer history; see, for example, Ref.~\cite{hartkamper1974foundations}. However, since the emergence of quantum information theory, the focus shifted to (finite-dimensional) composite systems, which brought renewed interest in the subject. It now is employed, 
for example, in the study of thermodynamics~\cite{barnum2010entropy,
chiribella2015entanglement,
chiribella2017microcanonical,
barnum2015entropy},
interference~\cite{
lee2017higher,
garner2018interferometric,
barnum2014higher,
dakic2014density,
lee2017higher,
barnum2017ruling,
horvat2020interference},
decoherence \cite{richens2017entanglement,
lee2018no,
scandolo2018possible,selby2017leaks}, 
computation~\cite{
lee2015computation,
barrett2019computational,
lee2015computation,
krumm2019quantum,
lee2016generalised,
lee2016deriving,
barnum2018oracles,
muller2012structure},
cryptography \cite{
sikora2018simple,
selby2018make,
sikora2019impossibility,
barnum2008nonclassicality,
lami2018ultimate,
DIQKD},
information processing~\cite{
bae2016structure,
barnum2011information,
GPT_Barrett,
JP17,
barnum2007generalized,
barnum2011information,
barnum2012teleportation,
barnum2013ensemble,
heinosaari2019no}, 
correlations~\cite{
czekaj2018bell,
barnum2010local,
czekaj2020correlations,
henson2014theory,
weilenmann2020analysing,
lami2018non,
cavalcanti2021witworld} and more~\cite{
masanes2019measurement,
galley2017classification,
galley2018any,
galley2020dynamics,
chiribella2011,
chiribella2014dilation,
chiribella2014distinguishability,
barnum2014local,
wilce2009four,
Royal-road,
barnum2016composites,
barnum2013symmetry,
wilce2011symmetry,
Masanes_2011,
masanes2013existence,
mueller2013three,aubrun2021entanglement}.
See \cite{garner2020characterization,plavala2021general,lami2018non} for recent reviews of the framework. 

In this work, we address the question of how one can describe {\em the particular states and effects that are accessible by a fixed experimental scenario} within a given GPT. As we will show, the mathematical structure that describes these accessible states and effects is not necessarily itself a GPT, but rather something that we term an {\em accessible GPT fragment}. 
These are more general than standard GPTs since (i) the accessible states and accessible effects in a given experimental set-up need not be tomographically complete for each other\footnote{Note that Gitton and Woods~\cite{Woods2020} introduce a notion of a `reduced space' based on quotienting relative to an equivalence relation wherein states are equivalent relative to the accessible effects and wherein effects are equivalent relative to the accessible states. This `reduced space' characterization of an experimental scenario is quite distinct from ours, and their resulting notion of classicality is quite distinct from generalized noncontextuality.}, and since (ii) there may be accessible states which are subnormalized 
and yet whose normalized counterparts 
 are not included in the fragment.\footnote{This latter point is similar to how in non-causal GPTs~\cite{chiribella2010probabilistic,dariano2014determinism} there are states that do not have a normalised counterpart~\cite{chiribella2010probabilistic}, but it occurs for a different reason; namely, the fact that a given experiment may not have access to repeat-until-success preparations.}

 Because accessible GPT fragments include all {\em and only} the states and effects accessible in a given experimental scenario, they are ideally suited to describing {\em relationships} between different experimental scenarios.
We give two specific applications of this idea: (i) to describe scenarios in which one measurement constitutes an {\em inefficient version} of another\footnote{That is, for each possible outcome of the original measurement, its inefficient counterpart has a nonzero probability of performing as expected, and a probability of failure (which is nonzero for at least one outcome), in which case it returns a `null' outcome.}, and (ii) to describe scenarios in which one measurement is a flag-convexification of another.  (Flag-convexification, introduced in Refs.~\cite{selby2021incompatibility,singh2021revealing}, is a process by which measurement settings are turned into measurement outcomes.)
  Moreover, as discussed in detail in Sec.~\ref{thryagtom}, we expect accessible GPT fragments to be a valuable tool for theory-agnostic tomography~\cite{mazurek2016experimental,mazurek2021experimentally,grabowecky2021experimentally}.
We also expect accessible GPT fragments to be critical for the development of a resource theory of generalized contextuality\footnote{Indeed, it was for this purpose that we originally conceived of them.}.

Next, we introduce an interesting equivalence relation that may hold between accessible GPT fragments (or, as a special case, between GPTs). We term this {\em cone equivalence}. Two accessible GPT fragments are said to be cone-equivalent if their state spaces define the same cone of states and their effect spaces define the same cone of effects. 

In both of the applications of accessible GPT fragments that we develop herein, cone equivalence plays a natural role: (i) In scenarios in which a given measurement device is implemented inefficiently, it follows that the accessible GPT fragments describing the original scenario and its inefficient counterpart are cone-equivalent. 
(ii) In scenarios where a measurement device is flag-convexified, it also follows that the accessible GPT fragments describing the original scenario and its flag-convexified counterpart are cone-equivalent. 

This concept of cone equivalence has also appeared indirectly in past work. 
 For instance, recognizing this fact explicitly allows us to 
restate one of the results of Ref.~\cite{wright2020general} as follows: 
A given GPT has a Gleason-like theorem (or equivalently, satisfies the no-restriction hypothesis for states) if and only if it is the unique GPT with the given state space which satisfies the full no-restriction hypothesis, {\em or one that 
 is cone-equivalent to that GPT}\footnote{Strictly, the GPTs that are cone-equivalent to an unrestricted~\cite{chiribella2010probabilistic} GPT correspond to the noisy-unrestricted GPTs of Ref.~\cite{wright2020general}, while their theorem singles out the slightly broader class of almost noisy-unrestricted GPTs which are those for which closure of the cones are equivalent to an unrestricted GPT. We will not focus on this distinction here as there is no operational distinction between the two classes.}. 

The notion of cone equivalence also plays an important role in understanding which GPTs and which accessible GPT fragments are classically explainable. We will now discuss classical explainability of GPTs and accessible GPT fragments, and then we will describe the role that cone equivalence plays.

One of the central questions in quantum foundations is how to make precise the senses in which our world cannot be explained classically.
In our view, the best foundational notion of classicality is the existence of a generalized-noncontextual ontological model\footnote{In fact, it is the slightly more general notion of Leibnizianity introduced in Ref.~\cite{schmid2020unscrambling}, but the distinction between the two will not be important here.}. This principle can be motivated by a useful methodological principle for theory construction dating back to Leibniz~\cite{Leibniz}.
Additionally, the existence of a  generalized-noncontextual ontological model for an operational theory coincides with two other independently motivated notions of classicality arising naturally in the study of GPTs~\cite{SchmidGPT,ShahandehGPT,schmid2020structure} and quantum optics~\cite{negativity,SchmidGPT,schmid2020structure}.
Generalized noncontextuality also emerges in the limit of sufficient decoherence~\cite{baldijao2021noncontextuality} or noise~\cite{PhysRevLett.115.110403,marvian2020inaccessible}. Furthermore, other key indicators of nonclassicality, such as violations of local causality~\cite{Bell} or observations of anomalous weak values under some conditions~\cite{AWV, KLP19}, are also instances of generalized contextuality. Finally, generalized contextuality has been proven to be a resource for information processing~\cite{POM,RAC,RAC2,Saha_2019,YK20}, computation~\cite{schmid2021only}, state discrimination~\cite{schmid2018contextual}, cloning~\cite{cloningcontext}, and metrology~\cite{contextmetrology}. 

The notion of generalized noncontextuality is most naturally defined within a framework of operational theories distinct from GPTs~\cite{gencontext}, wherein laboratory procedures that are operationally equivalent are not strictly equal. Such procedures are said to differ (only) by their {\em context}. Generalized noncontextuality, then, is the principle that an ontological representation of one's processes should be independent of their context. 
 Because the GPT that is associated to a given operational theory is obtained from the latter \blk by quotienting with respect to operational equivalences~\cite{chiribella2016quantum,schmid2020structure}, it follows that there are no contexts in a GPT, and hence no possibility of allowing for context-dependence\footnote{Indeed, one cannot even consider contexts in the sense of the Kochen-Specker theorem to be part of the fundamental structure of a GPT; see Appendix A of Ref.~\cite{schmid2020structure} for a more detailed justification of this claim.}.

It turns out~\cite{SchmidGPT},  however, that when mapping an operational
theory to a GPT by quotienting relative to operational equivalences, the constraint of explainability by
a generalized-noncontextual ontological model is mapped to the constraint of explainability by an ontological model.  Furthermore, it was shown in Ref.~\cite{SchmidGPT} that whether or not a GPT satisfies this constraint is determined by whether or not it satisfies a geometric criterion termed `simplex-embeddability'. 
 In particular, the state and effect spaces of such a GPT must be embeddable in those of a strictly classical (i.e., simplicial) GPT---one whose state space is a simplex and whose effect space is its dual.

 The fact that there is an equivalence between the existence of a generalised noncontextual ontological model for an operational theory, and the existence of a simplex embedding for the associated GPT, means that 
every motivation for taking generalized noncontextuality as one's notion of classicality for operational theories is just as much a motivation for taking simplex-embeddability as one's notion of classicality for GPTs. Naturally, then, we endorse it as the notion of classical explainability for GPTs.
 
Often, however, what one seeks to assess is not whether a particular \emph{operational theory} admits of a generalized-noncontextual model, but whether a particular \emph{operational scenario} or \emph{experiment} does. 
  Because such scenarios are  better described by  accessible GPT fragments, this raises the question of which accessible GPT fragments are classically explainable. It can be shown~\cite{simplex_forthcoming} (as a slight generalization of the results of Ref.~\cite{SchmidGPT}) that the geometric criterion implied by generalized noncontextuality for accessible GPT fragments is simply a version of the condition of simplex embeddability, adapted slightly so that it applies to accessible GPT fragments. 

In this work, we prove a useful result relating simplex embeddability and cone equivalence. In particular, we prove that two accessible GPT fragments which are cone-equivalent are either \emph{both} simplex-embeddable, or \emph{neither} is. In other words, two cone-equivalent accessible GPT fragments are either both nonclassical or neither is.

This result allows us to apply the notions of accessible GPT fragments and cone equivalence thereof to clarify the fundamental requirements for experimentally witnessing nonclassicality. In particular, we show that proofs of the failure of generalized noncontextuality can be found in which:
\begin{enumerate}[I.]
\item There is only a single measurement, and hence where
\begin{enumerate}
\item all measurements in the scenario are compatible, and
\item  no assumption of freedom of choice is required (i.e., no assumption that one can choose one's measurement setting in a manner that is independent of 
the choice of preparation and the hidden variable which in a putative classical model of the scenario mediates between the preparation and the measurement).\footnote{We take this to be the natural generalization (to prepare-measure scenarios) of the freedom of choice assumption commonly considered in Bell scenarios~\cite{wiseman2017causarum}.}
\end{enumerate}
\item One's measurement detectors have arbitrarily low (but nonzero) efficiency.
\end{enumerate}

It is straightforward to see each of these facts, given the preparatory work above. 
Consider first item I. For any scenario with multiple measurements, one can apply the flag-convexification procedure of Ref.~\cite{selby2021incompatibility} to turn the measurement settings into measurement outcomes, so that the resulting scenario has only a single measurement. As we will prove, this flag-convexified scenario is cone-equivalent to the original scenario, and hence either both are classically explainable, or both are not. Hence, this logic can be applied to \emph{any} proof of contextuality to obtain an alternative proof in the context of a new scenario with only a single measurement. This result can be viewed as the generalisation of the result of our companion paper, Ref.~\cite{selby2021incompatibility}, to arbitrary proofs of contextuality in arbitrary GPTs. 

Items I(a) and I(b) are immediate corollaries of Item I. Nonetheless, they are conceptually significant, insofar as they show a precise sense in which the experimental requirements and theoretical assumptions required to witness the failure of generalized noncontextuality are weaker than those required for witnessing the failure of Bell locality or Kochen-Specker noncontextuality.

 Finally, Item II follows from the fact that, as we will prove, a given prepare-measure scenario is cone-equivalent to a second scenario wherein the original measurements are implemented inefficiently; hence, either both are classically explainable, or both are not. 
Again, then, this logic can be applied to \emph{any} proof of contextuality to obtain an alternative proof in the context of a new scenario with detectors with arbitrarily low (but nonzero) efficiency. Unlike for experiments that aim to test Bell locality or Kochen-Specker noncontextuality, in experiments that aim to test generalized noncontextuality, then, there is no loophole involving detector inefficiency and hence no possibility of salvaging generalized noncontextuality via this loophole.

Moreover, we also show that proofs of the failure of generalized noncontextuality can be found in which there is a single source (as well as a single measurement). This shows that proofs of contextuality are possible where there are no setting variables at all; rather, such proofs rely only on passive observation of outcomes.

\section{Preliminaries}

The basis for this work is the framework of generalised probabilistic theories (GPTs). This framework enables one to describe a broad range of potential physical theories, based on the principle that ultimately any physical theory must be able to make probabilistic predictions about the outcomes of experiments. The GPT framework accommodates quantum theory and classical theory as special cases, but also allows for alternative physical theories such as quantum mechanics over the real numbers~\cite{hardy2012limited}, boxworld~\cite{GPT_Barrett}, or (the convex closure of) Spekkens' toy theory~\cite{spekkens2007evidence}. 

We will follow the diagrammatic approach to GPTs developed in Refs.~\cite{hardy2011reformulating,chiribella2010probabilistic}. In this approach, a given GPT $G$ consists of a collection of physical systems, which are diagrammatically denoted by wires, e.g.,
\beq
\begin{tikzpicture}
	\begin{pgfonlayer}{nodelayer}
		\node [style=none] (22) at (0, 0.75) {};
		\node [style=none] (23) at (0, -0.75) {};
		\node [style=right label] (24) at (0, -0.25) {$S$};
	\end{pgfonlayer}
	\begin{pgfonlayer}{edgelayer}
		\draw [qWire, in=270, out=90] (23.center) to (22.center);
	\end{pgfonlayer}
\end{tikzpicture}
}\ , \ \ %
\begin{tikzpicture}
	\begin{pgfonlayer}{nodelayer}
		\node [style=none] (22) at (0, 0.75) {};
		\node [style=none] (23) at (0, -0.75) {};
		\node [style=right label] (24) at (0, -0.25) {$S_1$};
	\end{pgfonlayer}
	\begin{pgfonlayer}{edgelayer}
		\draw [qWire] (23.center) to (22.center);
	\end{pgfonlayer}
\end{tikzpicture}
}%
\begin{tikzpicture}
	\begin{pgfonlayer}{nodelayer}
		\node [style=none] (22) at (0, 0.75) {};
		\node [style=none] (23) at (0, -0.75) {};
		\node [style=right label] (24) at (0, -0.25) {$S_2$};
	\end{pgfonlayer}
	\begin{pgfonlayer}{edgelayer}
		\draw [qWire] (23.center) to (22.center);
	\end{pgfonlayer}
\end{tikzpicture}
}\ , \ \ %
\begin{tikzpicture}
	\begin{pgfonlayer}{nodelayer}
		\node [style=none] (22) at (0, 0.75) {};
		\node [style=none] (23) at (0, -0.75) {};
		\node [style=right label] (24) at (0, -0.25) {$A$};
	\end{pgfonlayer}
	\begin{pgfonlayer}{edgelayer}
		\draw [cWire, in=270, out=90] (23.center) to (22.center);
	\end{pgfonlayer}
\end{tikzpicture}
}%
\begin{tikzpicture}
	\begin{pgfonlayer}{nodelayer}
		\node [style=none] (22) at (0, 0.75) {};
		\node [style=none] (23) at (0, -0.75) {};
		\node [style=right label] (24) at (0, -0.25) {$X$};
	\end{pgfonlayer}
	\begin{pgfonlayer}{edgelayer}
		\draw [cWire, in=270, out=90] (23.center) to (22.center);
	\end{pgfonlayer}
\end{tikzpicture}
}\,.
\eeq
The style of the wire denotes the type of system being represented.
Systems which describe classical degrees of freedom (e.g., those representing the controls and outcomes of experiments) are denoted by  thin wires, e.g., 
\beq
\begin{tikzpicture}
	\begin{pgfonlayer}{nodelayer}
		\node [style=none] (22) at (0, 0.75) {};
		\node [style=none] (23) at (0, -0.75) {};
		\node [style=right label] (24) at (0, -0.25) {$A$};
	\end{pgfonlayer}
	\begin{pgfonlayer}{edgelayer}
		\draw [cWire, in=270, out=90] (23.center) to (22.center);
	\end{pgfonlayer}
\end{tikzpicture}
}\ , \ \ %
\begin{tikzpicture}
	\begin{pgfonlayer}{nodelayer}
		\node [style=none] (22) at (0, 0.75) {};
		\node [style=none] (23) at (0, -0.75) {};
		\node [style=right label] (24) at (0, -0.25) {$X$};
	\end{pgfonlayer}
	\begin{pgfonlayer}{edgelayer}
		\draw [cWire, in=270, out=90] (23.center) to (22.center);
	\end{pgfonlayer}
\end{tikzpicture}
}\,,
\eeq
where the labels $A$ and $X$ are taken to be a (finite) set of possible classical states of the system, while systems which describe generic GPT systems are denoted by thick wires.

The GPT $G$ also describes a collection of processes such as:
\beq
\InputIfFileExists{Diagrams/generalProc1.tikz}{}{\input{./figures/Diagrams/generalProc1.tikz}}\ , \ \ %
\InputIfFileExists{Diagrams/generalProc2.tikz}{}{\input{./figures/Diagrams/generalProc2.tikz}}\ , \ \ %
\InputIfFileExists{Diagrams/generalProc3.tikz}{}{\input{./figures/Diagrams/generalProc3.tikz}}\ ,
\eeq
which describe the evolution of the physical systems. Processes can be wired together to form diagrams, such as
\beq
\InputIfFileExists{Diagrams/generalDiagram.tikz}{}{\input{./figures/Diagrams/generalDiagram.tikz}} \ ,
\eeq
which themselves describe other processes within the GPT. 

Processes with no inputs are known as states, and those with no outputs are known as effects. These are (respectively) depicted as
\beq
\begin{tikzpicture}
	\begin{pgfonlayer}{nodelayer}
		\node [style=none] (22) at (0, 0.75) {};
		\node [style=point] (23) at (0, -0.75) {$s$};
		\node [style=right label] (24) at (0, 0.25) {$S$};
	\end{pgfonlayer}
	\begin{pgfonlayer}{edgelayer}
		\draw [qWire] (23) to (22.center);
	\end{pgfonlayer}
\end{tikzpicture}
}\ , \ \ %
\begin{tikzpicture}
	\begin{pgfonlayer}{nodelayer}
		\node [style=none] (22) at (0, -0.75) {};
		\node [style=copoint] (23) at (0, 0.75) {$e$};
		\node [style=right label] (24) at (0, -0.25) {$S$};
	\end{pgfonlayer}
	\begin{pgfonlayer}{edgelayer}
		\draw [qWire] (23) to (22.center);
	\end{pgfonlayer}
\end{tikzpicture}
}.
\eeq

Processes with neither inputs nor outputs correspond to probabilities, e.g.,  
\beq
\begin{tikzpicture}
	\begin{pgfonlayer}{nodelayer}
		\node [style=scalar] (22) at (0, 0) {$p$};
	\end{pgfonlayer}
\end{tikzpicture}
} \ \ \in \ \ [0,1].
\eeq
Note, for example, that this means that the composition of a state with an effect, such as
\beq
\begin{tikzpicture}
	\begin{pgfonlayer}{nodelayer}
		\node [style=point] (22) at (0, -0.75) {$s$};
		\node [style=copoint] (23) at (0, 0.75) {$e$};
		\node [style=right label] (24) at (0, 0) {$S$};
	\end{pgfonlayer}
	\begin{pgfonlayer}{edgelayer}
		\draw [qWire] (23) to (22);
	\end{pgfonlayer}
\end{tikzpicture}
}\ \ = \ \ %
\begin{tikzpicture}
	\begin{pgfonlayer}{nodelayer}
		\node [style=scalar] (22) at (0, 0) {$p$};
	\end{pgfonlayer}
\end{tikzpicture}
} \,,
\eeq
is equal to a probability, namely that of observing the effect $e$ given the state $s$ of system $S$.

Each system comes with a specified discarding effect, termed the unit effect, denoted by $u_S$ and depicted as 
\beq
\begin{tikzpicture}
	\begin{pgfonlayer}{nodelayer}
		\node [style=none] (14) at (0, 0) {};
		\node [style=none] (15) at (0, -0.75) {};
		\node [style=right label] (21) at (0, -0.75) {$S$};
		\node [style=upground] (22) at (0, 0.25) {};
	\end{pgfonlayer}
	\begin{pgfonlayer}{edgelayer}
		\draw [qWire, in=-90, out=90] (15.center) to (14.center);
	\end{pgfonlayer}
\end{tikzpicture}
} \,.
\eeq
The unit effect represents ignoring a particular physical system (e.g., in quantum theory this would correspond to the operation of tracing out a system). 

Note that  a process that is deterministic---i.e., can be implemented without postselection---must satisfy the condition\footnote{This condition has been highlighted in Ref.~\cite{chiribella2010probabilistic}, which referred to it as the `causality' condition. We do not endorse this terminology for the reasons espoused in Ref.~\cite{schmid2020unscrambling}.} 
\beq\label{eq:causeq}
\InputIfFileExists{Diagrams/causality1.tikz}{}{\input{./figures/Diagrams/causality1.tikz}} \ \ = \ \ %
\begin{tikzpicture}
	\begin{pgfonlayer}{nodelayer}
		\node [style=none] (14) at (0, 0) {};
		\node [style=none] (15) at (0, -0.75) {};
		\node [style=right label] (21) at (0, -0.75) {$S_1$};
		\node [style=upground] (22) at (0, 0.25) {};
	\end{pgfonlayer}
	\begin{pgfonlayer}{edgelayer}
		\draw [qWire, in=-90, out=90] (15.center) to (14.center);
	\end{pgfonlayer}
\end{tikzpicture}
} \,.
\eeq
\emph{All} processes within the GPT must satisfy\footnote{This order on effects is defined by $e\leq f$ if and only if there exists an effect $e'$ such that $e+e'=f$.}
\beq
\InputIfFileExists{Diagrams/causality1.tikz}{}{\input{./figures/Diagrams/causality1.tikz}} \ \ \leq \ \ %
\begin{tikzpicture}
	\begin{pgfonlayer}{nodelayer}
		\node [style=none] (14) at (0, 0) {};
		\node [style=none] (15) at (0, -0.75) {};
		\node [style=right label] (21) at (0, -0.75) {$S_1$};
		\node [style=upground] (22) at (0, 0.25) {};
	\end{pgfonlayer}
	\begin{pgfonlayer}{edgelayer}
		\draw [qWire, in=-90, out=90] (15.center) to (14.center);
	\end{pgfonlayer}
\end{tikzpicture}
}\,,
\eeq
which represents the case of processes that may only happen with some non unit probability, such as the outcomes of measurements.
For such a probabilistic process, there exists another process
\beq
\InputIfFileExists{Diagrams/causality3.tikz}{}{\input{./figures/Diagrams/causality3.tikz}}
\eeq
such that
\beq
\InputIfFileExists{Diagrams/causality5.tikz}{}{\input{./figures/Diagrams/causality5.tikz}} \ \ := \ \ %
\InputIfFileExists{Diagrams/causality4.tikz}{}{\input{./figures/Diagrams/causality4.tikz}} \ \ + \ \ %
\InputIfFileExists{Diagrams/causality3.tikz}{}{\input{./figures/Diagrams/causality3.tikz}}
\eeq
is a valid deterministic process (i.e., it satisfies Eq.~\eqref{eq:causeq}).

Any system $S$ within the GPT has an associated space of states, denoted $\Omega_S$, which can be characterised as a convex set spanning\footnote{Recall that the span of a set of vectors in a real vector space is defined as $\mathsf{Span}[X] := \{\sum_i \lambda_i x_i | x_i \in X, \lambda_i \in \mathds{R} \}.$
} some finite-dimensional real vector space, which (with slight abuse of notation) we also denote as $S$.  
We will take the set  $\Omega_S$ to contain both normalized and unnormalized states for the system $S$, where the set of normalized states is defined as
\beq
\left\{ %
\begin{tikzpicture}
	\begin{pgfonlayer}{nodelayer}
		\node [style=none] (22) at (0, 0.75) {};
		\node [style=point] (23) at (0, -0.75) {$s$};
		\node [style=right label] (24) at (0, 0.25) {$S$};
	\end{pgfonlayer}
	\begin{pgfonlayer}{edgelayer}
		\draw [qWire] (23) to (22.center);
	\end{pgfonlayer}
\end{tikzpicture}
} \ \ \Bigg\vert \ \  %
\begin{tikzpicture}
	\begin{pgfonlayer}{nodelayer}
		\node [style=none] (14) at (0, 0.5) {};
		\node [style=point] (15) at (0, -0.75) {$s$};
		\node [style=right label] (21) at (0, 0) {$S$};
		\node [style=upground] (22) at (0, 0.75) {};
	\end{pgfonlayer}
	\begin{pgfonlayer}{edgelayer}
		\draw [qWire, in=-90, out=90] (15) to (14.center);
	\end{pgfonlayer}
\end{tikzpicture}
} \ \ = \ \ 1 \right\}\,.
\eeq
Geometrically, $\Omega_S$ is the convex hull of the set of normalized states and the origin (i.e., the state of normalization 0).

Any system $S$ within the GPT will also have an associated convex space of effects, denoted $\mathcal{E}_S$, living in the dual vector space to  
$S$ (that is, in the space of linear functionals on 
$S$
), denoted $S^*$.
We will also require a notion of duality for convex sets, which we will also denote by $*$ (following a standard abuse of notation).  Specifically, if $\Omega_S$ is a convex set, then the dual set $\Omega_S^*$ is defined as the set of effect vectors in the dual vector space that are \emph{logically possible} in the sense that they yield valid probabilities for all states, 
\begin{equation}\label{logicallypossibleeffects}
\Omega_S^* := \{  e | 0 \le e(s) \le 1, \forall s \in \Omega_S \}.
\end{equation}
 Clearly therefore,
  $\mathcal{E}_S \subseteq \Omega_S^*$.
When this inclusion relation is an equality, the GPT $G$ is said to satisfy the no-restriction hypothesis for effects~\cite{chiribella2010probabilistic}.  Both the origin of $S^*$,
termed the `zero' effect, and the unit effect $u_S$ must be elements of $\mathcal{E}_S$. 

Similarly, we can define the space of logically possible states as vectors in $S$ which yield valid probabilities for all effects,
\beq
\mathcal{E}_S^*:=\{s|0\leq e(s)\leq 1,\forall e \in\mathcal{E}_S\}.
\eeq
Clearly we also have that $\Omega_S \subseteq \mathcal{E}_S^*$. When this inclusion relation is an equality, the GPT $G$ is said to satisfy the no-restriction hypothesis for states~\cite{wright2020general}.

Probabilities are computed using the bilinear evaluation map
\begin{align}
B_S&:S^* \times S \to \mathds{R}\\
&:: (e,s) \mapsto e(s),
\end{align}
where if $e\in \mathcal{E}_S$ and $s\in \Omega_S$, we have that 
\beq \label{eq:standardgptProb}
B_S(e,s)=e(s) := %
\begin{tikzpicture}
	\begin{pgfonlayer}{nodelayer}
		\node [style=point] (22) at (0, -0.75) {$s$};
		\node [style=copoint] (23) at (0, 0.75) {$e$};
		\node [style=right label] (24) at (0, 0) {$S$};
	\end{pgfonlayer}
	\begin{pgfonlayer}{edgelayer}
		\draw [qWire] (23) to (22);
	\end{pgfonlayer}
\end{tikzpicture}
}\ \  \in [0,1]\,.
\eeq

An arbitrary system within a GPT has associated to it a tuple $(\Omega_S, \mathcal{E}_S, B_S, u_S)$. This tuple is itself often referred to as a GPT,  even though a full generalized probabilistic theory should also describe the possibilities for parallel composition of systems and transformations on systems. We will for simplicity make use of this terminology here, since the remainder of the article focuses only on prepare-measure scenarios with no parallel composition. 
It should be clear from context whether by `GPT' we are referring to this tuple or to the full compositional theory in which it lives.
Note also that the separate specification of the evaluation map and the unit is redundant, since the unit effect is uniquely determined by the prediction map. We have only included it for ease of comparison with the more general notion of accessible GPT fragments that we will introduce below.

\subsection{The notion of classicality for GPTs}\label{se:classicGPT}

As discussed in the introduction, a generalized probabilistic theory is classically-explainable if it is \emph{simplex-embeddable}, a concept first introduced in Ref.~\cite{SchmidGPT}. The following is a mild generalization of the definition appearing therein, in order to incorporate subnormalized states. (This generalization does not change the relationship between simplex embeddability and classical explainability, i.e., generalized noncontextuality.)

\begin{definition}\textbf{Simplex embeddability of standard GPTs.}\label{simplexembeddability}
\\
A GPT $G_S  =(\Omega_S,\mathcal{E}_S,B_S,u_S)$ is said to be simplex-embeddable if there exists some vector space dimension $n$ and a pair of linear maps $\iota:S\to \mathds{R}^n$ and $\kappa:S^*\to {\mathds{R}^n}^*$, such that for all $s\in \Omega_S$ and $e\in \mathcal{E}_S$ one has
\begin{align}
\iota(s) &\in \Delta_{\text{sub}}^n, \\
\kappa(e) &\in {{\Delta^n_{\text{sub}}}}^*, \\
\kappa(e)[\iota(s)]  &= B_S(e,s),\\
\kappa(u_S) &= \mathbf{1}_n.
\end{align}
where $\Delta^n_{\text{sub}}$ is the  simplex consisting of the convex hull of the zero vector and the standard basis vectors (those of the form $(1,0,\ldots,0), (0,1,\ldots, 0)$ etc.), ${\Delta^n_{\text{sub}}}^*$ is the dual convex set, and $\mathbf{1}_n$ is the vector $(1,1,...,1)$ of $n$ 1's.
\end{definition}
Note that the dimension $n$ of the simplex may be larger than the dimension of the given GPT system $S$. An explicit example of the necessity of this dimension mismatch is given for the case of standard GPTs in Appendix~D of Ref.~\cite{SchmidGPT}.

As motivated in the introduction, an operational theory is classically explainable if and only if it admits of a generalized-noncontextual representation. However, as discussed in detail in Ref.~\cite{SchmidGPT}, it does not make sense to ask if a generalized probabilistic theory (i.e., a \emph{quotiented} operational theory \cite{chiribella2010probabilistic,schmid2020structure}) admits of a generalized-noncontextual representation, since processes in a GPT do not have any contexts on which one's representation could possibly depend. However, Ref.~\cite{SchmidGPT} shows that an operational theory admits of a generalized-noncontextual ontological model if and only if the GPT which is obtained from it by quotienting can be embedded within a simplicial GPT via a linear map.

Hence, every motivation for taking realizability by a generalized noncontextual ontological model as one's notion of classical explainability for operational theories (like those listed in the introduction) is \emph{also} a motivation for taking simplex-embeddability as one's notion of classical explainability for GPTs. Furthermore, this notion of simplex-embeddability can be independently motivated as a notion of classical explainability, since simplicial GPTs are those wherein all possible measurements are compatible (and so are universally agreed to represent classical theories). Thus, the two notions of classical explainability mutually support one another. 

In this work, we are not primarily interested in assessing classicality of a GPT in its entirety. Rather, we are interested in assessing the classicality of the statistics that arise \emph{within that GPT for a particular prepare-measure scenario}, and so (in Definition~\ref{simplexembeddingfrag}) we will adapt the definition of simplex embeddability appropriately. 

We now introduce the scenarios of interest in more detail.

\subsection{Prepare-measure scenarios}

We now give a description of prepare-measure scenarios in generalized probabilistic theories.

Such scenarios are described by a set of sources and a set of measurements on some system $S$.
We index the possible sources by the set $X$, and label the classical outcome of the source by $A$:
\beq
\left\{%
\InputIfFileExists{Diagrams/QS.tikz}{}{\input{./figures/Diagrams/QS.tikz}}\right\}_{x\in X}.
\eeq
In a given run of the experiment, the source generates a classical outcome and a physical state of the system $S$, where one can think of the classical outcome as a flag which labels what preparation was done on $S$. (Note that a preparation procedure is a special case of a source where the classical output is unary.)
We index the set of measurements by the set $Y$, and label the classical outcome of the measurement by $B$:
\beq
\left\{%
\InputIfFileExists{Diagrams/meter.tikz}{}{\input{./figures/Diagrams/meter.tikz}} \right\}_{y\in Y}.
\eeq
These compose to give a set of joint distributions over source-and-measurement outcomes ($A$ and $B$), indexed by the possible choices of sources and measurements ($X$ and $Y$):
\beq
\left\{%
\InputIfFileExists{Diagrams/setofdist.tikz}{}{\input{./figures/Diagrams/setofdist.tikz}}\right\}_{x\in X, y\in Y}.
\eeq

For our purposes, however, it is more convenient to combine the set of sources together into a single process, and to combine the set of measurements together into a single process. That is, instead of working with a set of sources, we instead work with a single \emph{multisource}
\beq
\InputIfFileExists{Diagrams/QMS.tikz}{}{\input{./figures/Diagrams/QMS.tikz}}\,,
\eeq
which recovers each of the sources in the set by suitably choosing the setting variable, that is,
\beq
\InputIfFileExists{Diagrams/MStoS.tikz}{}{\input{./figures/Diagrams/MStoS.tikz}} \  = \ %
\InputIfFileExists{Diagrams/QS.tikz}{}{\input{./figures/Diagrams/QS.tikz}}
\eeq
for all $x\in X$.
Similarly, instead of working with a set of measurements, we work with a single \emph{multimeter}
\beq
\InputIfFileExists{Diagrams/QMM.tikz}{}{\input{./figures/Diagrams/QMM.tikz}}
\eeq
which recovers each of the measurements in the set by suitably choosing the setting variable; that is:
\beq
\InputIfFileExists{Diagrams/meter.tikz}{}{\input{./figures/Diagrams/meter.tikz}} \ \ = \ %
\InputIfFileExists{Diagrams/multimeter.tikz}{}{\input{./figures/Diagrams/multimeter.tikz}}
\eeq 
for all $y\in Y$. The composition of these yields a single stochastic map (or equivalently, a single conditional probability distribution):
\beq \label{eq:PMScenario}
\InputIfFileExists{Diagrams/PMScenario.tikz}{}{\input{./figures/Diagrams/PMScenario.tikz}}\,
\eeq
which gives the conditional probability  $p(ab|xy)$ of the pair of outcomes given the pair of inputs, where $a$ (resp.~$b,x,y$) denotes the value that the classical variable $A$ (resp.~$B,X,Y$) takes. This switch in perspective from sets of distributions to a single conditional distribution, has also been used within the sphere of Bell correlations, see Ref.~\cite{Bellreview}.

\subsection{Incompatibility of measurements}

A set of measurements (in our case, the set realized by a given multimeter) is said to be \emph{compatible} if there exists a single measurement that can simulate the statistics of any measurement in the set \cite{HMZ16},  otherwise the set is said to be \emph{incompatible}. This can be expressed simply in terms of the multimeter which realizes the measurements in the set, as follows.

\begin{definition}{\bf Compatibility} 
A set of measurements $\{M_y\}_{y\in Y}$ is said to be compatible if there exists a measurement $\mathcal{M}$ with some outcome set $Z$ and a controlled postprocessing $\mathcal{P}$ of $Z$ such that
\beq
\InputIfFileExists{Diagrams/parentMeasurementProcessed.tikz}{}{\input{./figures/Diagrams/parentMeasurementProcessed.tikz}} \ \ = \ \ %
\InputIfFileExists{Diagrams/QMM.tikz}{}{\input{./figures/Diagrams/QMM.tikz}}.
\eeq
\end{definition}
Note that a set comprised of a single measurement (hence, where $Y$ is the singleton set $\{*\}$) 
is necessarily compatible; one can simply take $\mathcal{M}=M=M_*$ and $\mathcal{P}=\mathds{1}$. 

\section{Accessible GPT fragments} \label{realizedgpts}

Typically, the GPT state and effect spaces associated to a given system are taken to include \emph{all} states and effects that can be physically realized on the system according to the theory. 
In contrast, we here want to describe \emph{only} those states and effects that can be realized \emph{within a given multisource-multimeter scenario}. In this section, we formalize the latter notion and term the resulting object an \emph{accessible GPT fragment}. We will refer to GPTs of the former (traditional) sort \emph{standard GPTs} when contrasting them with accessible GPT fragments. 

Note that a standard GPT is a special case of an accessible GPT fragment wherein the multisource and multimeter are able to access all physically possible states and effects.
The additional generality in the notion of an accessible GPT fragment will be critical for our later results, e.g., for describing physical situations involving flag-convexification (Section~\ref{firstexfc}) and detector inefficiency (Section~\ref{subsec:Inefficiencies}).
(We also expect it to be useful for defining the relevant notion of resourcefulness in the resource theory of failures of generalized noncontextuality.)

Recall that a GPT system $S$ has convex state space $\Omega_S$ and convex effect space $\mathcal{E}_S$.
A given multisource $P$ will only be able to prepare some convex subset of the states in $\Omega_S$, namely, those of the form
\beq
\begin{tikzpicture}
	\begin{pgfonlayer}{nodelayer}
		\node [style=none] (5) at (0, 1) {};
		\node [style=point] (6) at (0, -0.25) {$s$};
		\node [style=right label] (14) at (0, 0.75) {$S$};
	\end{pgfonlayer}
	\begin{pgfonlayer}{edgelayer}
		\draw [qWire] (6) to (5.center);
	\end{pgfonlayer}
\end{tikzpicture}
} =  %
\InputIfFileExists{Diagrams/OmegaP3.tikz}{}{\input{./figures/Diagrams/OmegaP3.tikz}} ,
\eeq
where $c$ is some substochastic comb, defined as follows. For some arbitrary  classical system  $Z$, $c$ can be decomposed as a bipartite subnormalized distribution $q$ over $X$ and $Z$, composed with some bipartite response function $r$ over $A$ and $Z$:
\beq \label{eq:crc}
\InputIfFileExists{Diagrams/classComb.tikz}{}{\input{./figures/Diagrams/classComb.tikz}} \ \ = \ \ %
\InputIfFileExists{Diagrams/classComb1.tikz}{}{\input{./figures/Diagrams/classComb1.tikz}}.
\eeq
We will refer to the set of such objects as $\mathsf{SubStochComb}$, or, in the case where $X$ is trivial, as $\mathsf{SubStochEff}$.  

Typically, the convex set of states a multisource $P$ can prepare 
 will be a strict subset of $\Omega_S$, and it may even be a lower-dimensional convex set. For example, one could have a multisource on a qubit that can only prepare the convex set corresponding to mixtures of the states $\ket{0}$ and $\ket{1}$. It will therefore be convenient to view the state space for the accessible GPT fragment as living in the vector space spanned by the accessible states, rather than the original GPT vector space.
We denote this span by 
\beq
S_P := \mathsf{Span}\left[\left\{%
\InputIfFileExists{Diagrams/OmegaP3.tikz}{}{\input{./figures/Diagrams/OmegaP3.tikz}}\middle| c \in \mathsf{SubStochComb} \right\}\right]\,.
\eeq 

Taking this view requires introducing two maps which have little physical significance, but are useful for proper bookkeeping.
Since $S_P$ is a subspace of $S$, one may define a projection map $S\to S_P$ and an inclusion map $S_P\to S$. We denote these diagrammatically by: 
\begin{align}
\projP \, &: \, S\to S_P\, :: %
\begin{tikzpicture}
	\begin{pgfonlayer}{nodelayer}
		\node [style=none] (5) at (0, 0.5) {};
		\node [style=point] (6) at (0, -0.5) {$v$};
		\node [style=right label] (7) at (0, 0.25) {$S$};
	\end{pgfonlayer}
	\begin{pgfonlayer}{edgelayer}
		\draw [qWire] (6) to (5.center);
	\end{pgfonlayer}
\end{tikzpicture}
} \mapsto  %
\InputIfFileExists{Diagrams/projP2.tikz}{}{\input{./figures/Diagrams/projP2.tikz}},\\
\incP \, &: \, S_P\to S\, :: %
\begin{tikzpicture}
	\begin{pgfonlayer}{nodelayer}
		\node [style=none] (5) at (0, 0.5) {};
		\node [style=point] (6) at (0, -0.5) {$u$};
		\node [style=right label] (7) at (0, 0.25) {$S_P$};
	\end{pgfonlayer}
	\begin{pgfonlayer}{edgelayer}
		\draw [qWire] (6) to (5.center);
	\end{pgfonlayer}
\end{tikzpicture}
} \mapsto  %
\InputIfFileExists{Diagrams/incP2.tikz}{}{\input{./figures/Diagrams/incP2.tikz}}.
\end{align}
 Note that the inclusion map followed by the projection map is the identity,
\beq
\InputIfFileExists{Diagrams/proj1.tikz}{}{\input{./figures/Diagrams/proj1.tikz}} = %
\begin{tikzpicture}
	\begin{pgfonlayer}{nodelayer}
		\node [style=none] (2) at (0, 1) {};
		\node [style=none] (3) at (0, -1) {};
		\node [style=right label] (4) at (0, -0.25) {$S_P$};
	\end{pgfonlayer}
	\begin{pgfonlayer}{edgelayer}
		\draw [qWire] (2.center) to (3.center);
	\end{pgfonlayer}
\end{tikzpicture}
},
\eeq
and the projection map followed by the inclusion map is the identity \emph{on the accessible states},
\beq
\InputIfFileExists{Diagrams/proj5.tikz}{}{\input{./figures/Diagrams/proj5.tikz}} = %
\InputIfFileExists{Diagrams/QMS2.tikz}{}{\input{./figures/Diagrams/QMS2.tikz}}. 
\eeq

 Hence, we define the state space for the accessible GPT fragment as 
\beq
\Omega_P := \left\{ %
\begin{tikzpicture}
	\begin{pgfonlayer}{nodelayer}
		\node [style=none] (5) at (0, 0.5) {};
		\node [style=point] (6) at (0, -0.75) {$s$};
		\node [style=right label] (14) at (0, 0.25) {$S_P$};
	\end{pgfonlayer}
	\begin{pgfonlayer}{edgelayer}
		\draw [qWire] (6) to (5.center);
	\end{pgfonlayer}
\end{tikzpicture}
} = 
\InputIfFileExists{Diagrams/OmegaPn.tikz}{}{\input{./figures/Diagrams/OmegaPn.tikz}} \middle|  c \in \mathsf{SubStochComb} \right\}.
\eeq 
It is not hard to see that this definition of the state space implies constraints on the geometry of the set $\Omega_P$. Fig.~\ref{fig:tofragment} provides an illustrative example and some discussion of these constraints.

We now turn to the effects. The effects in $\mathcal{E}_S$ which can be achieved using the multimeter $M$ are those of the form
\beq
\begin{tikzpicture}
	\begin{pgfonlayer}{nodelayer}
		\node [style=none] (5) at (0, -1) {};
		\node [style=copoint] (6) at (0, 0.25) {$e$};
		\node [style=right label] (14) at (0, -0.75) {$S$};
	\end{pgfonlayer}
	\begin{pgfonlayer}{edgelayer}
		\draw [qWire] (6) to (5.center);
	\end{pgfonlayer}
\end{tikzpicture}
} = %
\InputIfFileExists{Diagrams/calEM.tikz}{}{\input{./figures/Diagrams/calEM.tikz}},
\eeq
where $d$ is a substochastic comb. Again, these will form a convex subset of $\mathcal{E}_S$ which will span a subspace of $S^*$. We denote the subspace spanned by this set as $S_M^*$, and we define a (contravariant) projection map
\beq
\projM \, : \, S^*\to S_M^* :: %
\begin{tikzpicture}
	\begin{pgfonlayer}{nodelayer}
		\node [style=none] (5) at (0, -0.75) {};
		\node [style=copoint] (6) at (0, 0.25) {$v'$};
		\node [style=right label] (7) at (0, -0.5) {$S$};
	\end{pgfonlayer}
	\begin{pgfonlayer}{edgelayer}
		\draw [qWire] (6) to (5.center);
	\end{pgfonlayer}
\end{tikzpicture}
} \mapsto %
\InputIfFileExists{Diagrams/projM2.tikz}{}{\input{./figures/Diagrams/projM2.tikz}} \,,
\eeq
and a (contravariant) inclusion map 
\beq
\incM \, : \, S_M^* \to S^* :: %
\begin{tikzpicture}
	\begin{pgfonlayer}{nodelayer}
		\node [style=none] (5) at (0, -0.5) {};
		\node [style=copoint] (6) at (0, 0.5) {$u'$};
		\node [style=right label] (7) at (0, -0.25) {$S_M$};
	\end{pgfonlayer}
	\begin{pgfonlayer}{edgelayer}
		\draw [qWire] (6) to (5.center);
	\end{pgfonlayer}
\end{tikzpicture}
} \mapsto %
\InputIfFileExists{Diagrams/incM2.tikz}{}{\input{./figures/Diagrams/incM2.tikz}}\,  ,
\eeq
where these are read top-to-bottom, as we view these as acting contravariantly (i.e., on the dual spaces). 
The inclusion map followed by the projection map is the identity:
\beq
\InputIfFileExists{Diagrams/proj2.tikz}{}{\input{./figures/Diagrams/proj2.tikz}} = %
\begin{tikzpicture}
	\begin{pgfonlayer}{nodelayer}
		\node [style=none] (2) at (0, 1) {};
		\node [style=none] (3) at (0, -1) {};
		\node [style=right label] (4) at (0, -0.25) {$S_M$};
	\end{pgfonlayer}
	\begin{pgfonlayer}{edgelayer}
		\draw [qWire] (2.center) to (3.center);
	\end{pgfonlayer}
\end{tikzpicture}
}\,.
\eeq

Furthermore, the projection map followed by the inclusion map is the identity \emph{on the accessible effects}:
\beq
\InputIfFileExists{Diagrams/proj6.tikz}{}{\input{./figures/Diagrams/proj6.tikz}} = %
\InputIfFileExists{Diagrams/QMM2.tikz}{}{\input{./figures/Diagrams/QMM2.tikz}}.
\eeq

We can then define the effect space for the accessible GPT fragment as
\beq
\mathcal{E}_M := \left\{ %
\begin{tikzpicture}
	\begin{pgfonlayer}{nodelayer}
		\node [style=none] (5) at (0, -0.75) {};
		\node [style=copoint] (6) at (0, 0.5) {$e$};
		\node [style=right label] (14) at (0, -0.5) {$S_M$};
	\end{pgfonlayer}
	\begin{pgfonlayer}{edgelayer}
		\draw [qWire] (6) to (5.center);
	\end{pgfonlayer}
\end{tikzpicture}
}= 
\InputIfFileExists{Diagrams/camEMn.tikz}{}{\input{./figures/Diagrams/camEMn.tikz}}  \middle| d \in \mathsf{SubStochComb}  \right\}.
\eeq 

It is not hard to see that this definition of the effect space implies constraints on the geometry of the set $\mathcal{E}_M$, analogous to those on the state spaces (illustrated in Fig.~\ref{fig:tofragment}).

 We can now observe the first key formal distinction between accessible GPT fragments and GPTs in the traditional sense: the vector space $S_P$ associated to the state space is not necessarily the same as the vector space $S_M$ associated to the effect space. In a traditional GPT, these two vectors spaces are equal, as follows from the assumption that the states are tomographically complete for the effects and vice-versa. This can fail to be the case for an accessible GPT fragment because we are not requiring that the multisource and multimeter need to be tomographically complete for one another.

Given that the states and effects are no longer associated to the same vector space, we can no longer compute probabilities simply by composition, as in Eq.~\eqref{eq:standardgptProb}. In general, the sequential composition of
\beq
\begin{tikzpicture}
	\begin{pgfonlayer}{nodelayer}
		\node [style=none] (5) at (0, -0.75) {};
		\node [style=copoint] (6) at (0, 0.5) {$e$};
		\node [style=right label] (14) at (0, -0.5) {$S_M$};
	\end{pgfonlayer}
	\begin{pgfonlayer}{edgelayer}
		\draw [qWire] (6) to (5.center);
	\end{pgfonlayer}
\end{tikzpicture}
} \in \mathcal{E}_M \ \ \text{and} \ \ %
\begin{tikzpicture}
	\begin{pgfonlayer}{nodelayer}
		\node [style=none] (5) at (0, 0.5) {};
		\node [style=point] (6) at (0, -0.75) {$s$};
		\node [style=right label] (14) at (0, 0.25) {$S_P$};
	\end{pgfonlayer}
	\begin{pgfonlayer}{edgelayer}
		\draw [qWire] (6) to (5.center);
	\end{pgfonlayer}
\end{tikzpicture}
} \in \Omega_P
\eeq
is not even defined, as their types do not always match. 
To compute probabilities, then, one first embeds each system type ($S_P$ and $S_M$) back into $S$ using their respective inclusion maps, and then one computes the probability using the probability rule in the original GPT.

Explicitly, the probability rule for an accessible GPT fragment is given by the bilinear map
\begin{align}
{B}_{PM} &: S_M^*\times S_P \to \mathds{R} \\
 &:: \left(%
\begin{tikzpicture}
	\begin{pgfonlayer}{nodelayer}
		\node [style=none] (5) at (0, -0.75) {};
		\node [style=copoint] (6) at (0, 0.5) {$e$};
		\node [style=right label] (14) at (0, -0.5) {$S_M$};
	\end{pgfonlayer}
	\begin{pgfonlayer}{edgelayer}
		\draw [qWire] (6) to (5.center);
	\end{pgfonlayer}
\end{tikzpicture}
} ,%
\begin{tikzpicture}
	\begin{pgfonlayer}{nodelayer}
		\node [style=none] (5) at (0, 0.5) {};
		\node [style=point] (6) at (0, -0.75) {$s$};
		\node [style=right label] (14) at (0, 0.25) {$S_P$};
	\end{pgfonlayer}
	\begin{pgfonlayer}{edgelayer}
		\draw [qWire] (6) to (5.center);
	\end{pgfonlayer}
\end{tikzpicture}
}\right) \mapsto %
\InputIfFileExists{Diagrams/probProj.tikz}{}{\input{./figures/Diagrams/probProj.tikz}}\,. \label{eq:GPTProb}
\end{align}
where if $e \in \mathcal{E}_M $ and $ s\in  \Omega_P$ we have that
\beq
{B}_{PM}(e,s) =  %
\InputIfFileExists{Diagrams/probProj.tikz}{}{\input{./figures/Diagrams/probProj.tikz}} \in [0,1].
\eeq
\begin{figure*}[htb!]
\[%
\InputIfFileExists{Diagrams/Fig1New.tikz}{}{\input{./figures/Diagrams/Fig1New.tikz}}\]
\caption[.]{ 
A series of figures illustrating how the state space of an accessible GPT fragment arises from its associated multisource $P$.  In these examples, the grey region demarcates the full set of \emph{normalised} states in the underlying GPT within which the fragment is embedded.  In (a),  we depict the atomic preparations---namely, the state vectors describing the output of the multisource for a fixed setting value $x\in X$ and a fixed outcome $a\in A$ of the multisource. In this example, setting $x=0$ leads to two possible outcomes, with the two associated (subnormalized) state vectors depicted in red; similarly for $x=1$ in green. Setting $x=2$ has only a single outcome, so the corresponding state vector -- depicted in blue --  is normalized and lies in the  grey region.  (b) we show all of these atomic preparations on a single picture,  (c) Even for a fixed setting value, one can access more general preparations by considering outcomes described by a generic effect $e$; that is, by considering outcomes arising from arbitrary postprocessings of the atomic outcomes. In doing so, one can access all the state vectors in the red (for $x=0$), green (for $x=1$), and blue (for $x=2$) regions. Mathematically, these regions are zonotopes. (A zonotope is a convex set constructed by taking some set of vectors $v_i$ and then taking the set of all vectors that can be reached by linear combinations of the form $\sum_i \alpha_i v_i$ such that $\alpha_i \in [0,1]$.) Note that every such zonotope has a single vector in the underlying normalized state space (the grey region), and that this vector is generated by the corresponding diagram when $e$ is the unit effect (so that one effectively ignores the outcome of the multisource).  (d) The most general states that are achievable for a given multisource can be accessed by considering arbitrary pre-and-postprocessings, depicted here by the comb $c$ in the diagram. The accessible GPT fragment's state space corresponds to the displayed polytope, which is the convex hull of the zonotopes from (b). This example is illustrative of the generic case. That is, the accessible state space is always given by the convex hull of a set of zonotopes each contained within the full state space, where the maximal element of each zonotope lies in the hyperplane of normalized states. Note, moreover, that any such geometry can be realized by some multisource within the full GPT. Similarly, the accessible effect space is the convex hull of a set of zonotopes each contained within the full effect space where, in this case, the maximal elements are all simply the unit effect. Again we find that any such geometry can be realized by some multimeter within the full GPT.
Note that in this example, the subset of normalized states in the fragment forms a strict subset of the normalized states in the fundamental underlying GPT.
}\label{fig:tofragment}
\end{figure*}

Finally, the unit effect $u_M$ for the accessible GPT fragment  is defined as the unit $u_S$ for the GPT system $S$ preceded by projection into the subspace $S_M$. Note that $u_M$ is the effect that is realized for any deterministic use of the multimeter $M$, wherein one coarse-grains over all outcomes $B$ (and chooses the value of $Y$ according to any normalized probability distribution $p$). In other words, one has
\beq
\begin{tikzpicture}
	\begin{pgfonlayer}{nodelayer}
		\node [style=none] (5) at (-0.5, 0) {};
		\node [style=none] (6) at (-0.5, -1.5) {};
		\node [style=right label] (14) at (-0.5, -1.25) {$S_M$};
		\node [style=upground] (15) at (-0.5, 0.25) {};
	\end{pgfonlayer}
	\begin{pgfonlayer}{edgelayer}
		\draw [qWire] (6.center) to (5.center);
	\end{pgfonlayer}
\end{tikzpicture}
} := \ \ %
\InputIfFileExists{Diagrams/GPTUnit1.tikz}{}{\input{./figures/Diagrams/GPTUnit1.tikz}} \ \ = \ \ %
\InputIfFileExists{Diagrams/GPTUnit.tikz}{}{\input{./figures/Diagrams/GPTUnit.tikz}}. \label{defnunit}
\eeq
This last equality follows from the fact that, using Eq.~\eqref{eq:causeq},
\beq
\InputIfFileExists{Diagrams/GPTUnitn.tikz}{}{\input{./figures/Diagrams/GPTUnitn.tikz}}\ \  =\ \ %
\InputIfFileExists{Diagrams/GPTUnit4n.tikz}{}{\input{./figures/Diagrams/GPTUnit4n.tikz}}\ \  =\ \  %
\begin{tikzpicture}
	\begin{pgfonlayer}{nodelayer}
		\node [style=none] (5) at (-0.5, 0) {};
		\node [style=none] (6) at (-0.5, -1.25) {};
		\node [style=right label] (14) at (-0.5, -0.5) {$S$};
		\node [style=upground] (15) at (-0.5, 0.25) {};
	\end{pgfonlayer}
	\begin{pgfonlayer}{edgelayer}
		\draw [qWire] (6.center) to (5.center);
	\end{pgfonlayer}
\end{tikzpicture}
}  .
\eeq

Note that in an accessible GPT fragment (just as in a standard GPT), it holds that for any effect $e \in \mathcal{E}_M$, there exists another effect $e^\perp \in \mathcal{E}_M$ such that 
\begin{align}\label{eq:MeasCompletion}
e+e^\perp = u_M.
\end{align}
 This is proven in Appendix~\ref{proofcomplem}.

\begin{definition}{\bf Accessible GPT fragment}\label{accessibleGPTdefn}
\\
We call the tuple $G_{PM}:=(\Omega_P, \mathcal{E}_M, {B}_{PM},u_M)$ (whose elements are defined above) the \emph{accessible GPT fragment} associated to the multisource-multimeter pair $(P,M)$.
\end{definition}

The second key formal distinction between standard GPTs and accessible GPT fragments is that the latter may contain subnormalized states \emph{whose normalized counterparts are not contained in the accessible GPT fragment}. This happens every time a given state appears \emph{only} in the context of multisources with more than one outcome arising with nonzero probability. In a standard GPT, by contrast, the normalized version of every subnormalized GPT state is also a physically valid state (and hence included in the GPT).  This is because for 
every subnormalized state in one's GPT state space, it is physically possible to generate a normalized version of that state, e.g., by repeat-until-success preparation procedures~\cite{chiribella2010probabilistic}.\footnote{The fact that there is no such repeat-until success procedure for measurements leads to an interesting asymmetry for standard GPTs, namely, that satisfaction of the no-restriction hypothesis for effects implies satisfaction of the no-restriction hypothesis for states, while the converse is not true. For accessible GPT fragments, the implication does not hold in either direction.} 
Such procedures are not possible within a multisource-multimeter scenario, which is why accessible GPT fragments may exclude the normalized counterparts of some subnormalized states in the fragment. One can see this in the example we depict in Fig.~\ref{fig:tofragment}.

\subsection{The notion of classicality for accessible GPT fragments}
As discussed in Section~\ref{se:classicGPT}, a GPT is 
 classically-explainable 
if and only if it is simplex-embeddable. Similarly, an accessible GPT fragment is 
 classically-explainable 
 if and only if it is simplex-embeddable. In the latter case, however, the notion of simplex-embeddability must be adapted from that given in Ref.~\cite{SchmidGPT} (or above, in Definition~\ref{simplexembeddability}) for standard GPTs to the appropriate notion for accessible GPT fragments. In particular, it must be defined for subnormalized state spaces
 (rather than for normalized state spaces, as in Ref.~\cite{SchmidGPT}). 

Simplex-embeddability of accessible GPT fragments is defined as follows.
\begin{definition} \textbf{Simplex-embeddability  of accessible GPT fragments.} \label{simplexembeddingfrag}
Consider a given multisource-multimeter pair $(P,M)$ and the associated accessible GPT fragment  $G_{PM}=(\Omega_P,\mathcal{E}_M,{B}_{PM},u_M)$.
$G_{PM}$ is said to be simplex-embeddable if there exists some $n$ and a pair of linear maps $\iota:S_P\to \mathds{R}^n$ and $\kappa:{S_M}^*\to {\mathds{R}^n}^*$ such that for all $s\in \Omega_P$ and $e\in \mathcal{E}_M$, one has
\begin{align}
\iota(s) &\in \Delta_{\text{sub}}^n, \\
\kappa(e) &\in {{\Delta^n_{\text{sub}}}}^*, \\
\kappa(e)[\iota(s)]  &= {B}_{PM}(e,s),\\
\kappa(u_M) &= \mathbf{1}_n,
\end{align}
where $\Delta^n_{\text{sub}}$ is the  simplex consisting of the convex hull of the zero vector and the standard basis vectors (those of the form $(1,0,\ldots,0), (0,1,\ldots, 0)$ etc.), ${\Delta^n_{\text{sub}}}^*$ is the dual convex set, and $\mathbf{1}_n$ is the vector $(1,1,...,1)$ of $n$ 1's.
\end{definition}
Note that the dimension $n$ of the  simplex 
may be larger than the dimension of the given accessible GPT fragment.

Note that if a GPT $G_S$ is simplex-embeddable in the sense of Definition~\ref{simplexembeddability}, 
then any multisource-multimeter pair $(P,M)$ in $G_S$
 defines an accessible GPT fragment $G_{PM}$ that is simplex-embeddable in the sense of Definition~\ref{simplexembeddingfrag}. It follows that if a given multisource-multimeter pair $(P,M)$ in $G_S$  is not classically-explainable, then one can infer that the GPT $G_S$ fails to be classically-explainable.

\section{Cone-equivalent accessible GPT fragments }

We now introduce a kind of relationship between accessible GPT fragments (or between standard GPTs), which we term cone equivalence. Our definition draws inspiration from the idea of \emph{noisy unrestricted GPTs} from Ref.~\cite{wright2020general},  and we recently were made aware that the analogous concept of projective equivalence has been used, for example, in the proofs of Ref.~\cite{aubrun2021entangleability};  to our knowledge, however, cone equivalence has not previously been studied in its own right. After defining the notion, we will give two distinct experimental scenarios in which cone-equivalent accessible GPT fragments arise. (The cone equivalence in these two examples will then be central to proving our results about witnessing nonclassicality.)

\begin{definition}\textbf{Cone-equivalent accessible GPT fragments.}\label{def:CE}\\
A pair of accessible GPT fragments with the same probability rule and unit, $G=(\Omega,\mathcal{E},{B},u)$ and $G'=(\Omega',\mathcal{E}',{B},u)$, are said to be cone equivalent if and only if 
\begin{align}
    \mathsf{Cone}[\Omega] &= \mathsf{Cone}[\Omega']\label{eq:states_cone_eq}\\
    \mathsf{Cone}[\mathcal{E}] &= \mathsf{Cone}[\mathcal{E}'],
\end{align}
where the \emph{cone} associated to some set of vectors $X$ is defined as
\beq
\mathsf{Cone}[X] := \{ \lambda x | x\in X, \lambda \in \mathds{R}^+ \}.
\eeq
\end{definition}

Figure~\ref{cegpts} shows a simple example of two cone-equivalent accessible GPT fragments. 
\begin{figure*}[htb!]
\[%
\InputIfFileExists{Diagrams/Fig2New.tikz}{}{\input{./figures/Diagrams/Fig2New.tikz}}\]
\caption{
Here we depict two distinct three-dimensional accessible GPT fragment state spaces which define identical cones. (a) The first is the inner polytope in dark violet-grey constructed in Fig.~\ref{fig:tofragment}; (b) the second is the polytope in orange, which happens to contain the dark violet-grey one.(c) The \emph{cones} defined by these two state spaces coincide. The extremal rays of the cone common to both polytopes are depicted by the protruding red arrows.  As in Fig.~\ref{fig:tofragment}, the grey region demarcates the  convex set of \emph{normalized} states for the full GPT.  The black triangle lies within this region, as it corresponds to the set of normalized states in the inner polytope.
}
\label{cegpts}
\end{figure*}

In the next two subsections, we will give two physically-motivated scenarios in which cone-equivalent accessible GPT fragments arise naturally. This fact is then central to all of the results that we derive regarding noncontextuality. 

\begin{rem}\label{rem:EquivalentGPT}
One can also define a broader notion of cone equivalence which, in particular, does not limit the scope to accessible GPT fragments with the same probability rule and unit. 
This is done by leveraging the natural notion of equivalence for accessible GPT fragments. (As an example, $(\Omega,\mathcal{E},{B},u)$ is equivalent to $(T(\Omega),T'(\mathcal{E}),{B}\circ {T'}^{-1}\times T^{-1}, T'(u))$ for any reversible linear maps $T$ and $T'$.) That is, one can define a pair of accessible GPT fragments to be cone-equivalent if the pair is \emph{equivalent to a cone-equivalent pair} as per Def.~\ref{def:CE}. Our main technical result (Theorem~\ref{maintechnicalres}) holds under this more permissive notion of cone equivalence. However, this level of generality is unnecessary for any of our examples or key corollaries, and hence we omit it for simplicity.
\end{rem}

As discussed in the introduction, the notion of cone-equivalence can be used to shed light on the results of Ref.~\cite{wright2020general}. Another example of this is the following.
\begin{rem}\label{remarkonWrightetal}
As a simple extension of arguments in Ref.~\cite{wright2020general}, one can show that all cone-equivalent accessible GPT fragments have the same space of \emph{logically possible} states and have the same space of \emph{logically possible} effects.
 This fact disproves a common belief that whenever one shrinks the effect space of a given GPT, the  space of logically possible states 
  will necessarily grow. This is because if one shrinks the effect space by  merely rescaling some of the effects, the cone defined by the effect space is unchanged,  and the space of logically possible states depends only on this cone.
\end{rem}

\subsection{Example 1: flag-convexified scenarios} \label{firstexfc}

Our first example of cone-equivalent accessible GPT fragments arises from a process that we call \emph{flag-convexification}~\cite{selby2021incompatibility,singh2021revealing}. Suppose that we have a prepare and measure scenario described by a multisource and multimeter as in Eq.~\eqref{eq:PMScenario}. Now consider a new scenario in which one does not directly choose the setting variables, but rather allows each setting's value to be selected probabilistically according to some fixed, full-support distribution $\mu$ (for the multisource) or $\nu$ (for the multimeter). These random variables are moreover each copied to some flag variable which is treated as an additional output to the multisource (or multimeter).

More formally, flag-convexification of the multisource and multimeter 
gives a new scenario with a single source and a single meter (the white dot represents a copy operation),
\begin{align}
\InputIfFileExists{Diagrams/fcSource.tikz}{}{\input{./figures/Diagrams/fcSource.tikz}} &:= %
\InputIfFileExists{Diagrams/flagConvNonUniP.tikz}{}{\input{./figures/Diagrams/flagConvNonUniP.tikz}} , \text{ and}\\ %
\InputIfFileExists{Diagrams/fcMeter.tikz}{}{\input{./figures/Diagrams/fcMeter.tikz}} &:=%
\InputIfFileExists{Diagrams/flagConvNonUniM.tikz}{}{\input{./figures/Diagrams/flagConvNonUniM.tikz}},
\end{align}
to obtain the flag-convexified source-meter pair:
\beq
\left( %
\InputIfFileExists{Diagrams/fcSource.tikz}{}{\input{./figures/Diagrams/fcSource.tikz}} ,  %
\InputIfFileExists{Diagrams/fcMeter.tikz}{}{\input{./figures/Diagrams/fcMeter.tikz}}\right).
\eeq
 
The accessible GPT fragment $G_{P^\mu M^\nu}=(\Omega_{P^\mu}, \mathcal{E}_{M^\nu}, {B}_{P^\mu M^\nu}, u_{M^\nu})$ realized by this flag-convexified source-meter pair is defined by 
\begin{align}\label{eq:fcGPT}
\Omega_{P^\mu} &:= \left\{%
\begin{tikzpicture}
	\begin{pgfonlayer}{nodelayer}
		\node [style=none] (5) at (0, 0.5) {};
		\node [style=point] (6) at (0, -0.75) {$s$};
		\node [style=right label] (14) at (0, 0.25) {$S_{P^\mu}$};
	\end{pgfonlayer}
	\begin{pgfonlayer}{edgelayer}
		\draw [qWire] (6) to (5.center);
	\end{pgfonlayer}
\end{tikzpicture}
} =
% := \tikzfig{Diagrams/stateProjc} \middle| \tikzfig{Diagrams/stateP1} = 
%
\InputIfFileExists{Diagrams/fcnuOmegacn.tikz}{}{\input{./figures/Diagrams/fcnuOmegacn.tikz}} \middle|  c' \in \mathsf{SubStochEff}  \right\},  \\
\mathcal{E}_{M^\nu} &:=  \left\{%
\begin{tikzpicture}
	\begin{pgfonlayer}{nodelayer}
		\node [style=none] (5) at (0, -0.75) {};
		\node [style=copoint] (6) at (0, 0.5) {$e$};
		\node [style=right label] (14) at (0, -0.5) {$S_{M^\nu}$};
	\end{pgfonlayer}
	\begin{pgfonlayer}{edgelayer}
		\draw [qWire] (6) to (5.center);
	\end{pgfonlayer}
\end{tikzpicture}
} =
%:= \tikzfig{Diagrams/effectProjc} \middle| \tikzfig{Diagrams/effectM1} = 
%
\InputIfFileExists{Diagrams/fcnucalEcn.tikz}{}{\input{./figures/Diagrams/fcnucalEcn.tikz}}\middle|  d' \in \mathsf{SubStochEff}  \right\},  \\
{B}_{P^\mu M^\nu} &: {S_{M^\nu}}^*\times S_{P^\mu} \to \mathds{R}, \nonumber \\
  &::\left(%
\begin{tikzpicture}
	\begin{pgfonlayer}{nodelayer}
		\node [style=none] (5) at (0, 0.5) {};
		\node [style=point] (6) at (0, -0.75) {$s$};
		\node [style=right label] (14) at (0, 0.25) {$S_{P^\mu}$};
	\end{pgfonlayer}
	\begin{pgfonlayer}{edgelayer}
		\draw [qWire] (6) to (5.center);
	\end{pgfonlayer}
\end{tikzpicture}
} , %
\begin{tikzpicture}
	\begin{pgfonlayer}{nodelayer}
		\node [style=none] (5) at (0, -0.75) {};
		\node [style=copoint] (6) at (0, 0.5) {$e$};
		\node [style=right label] (14) at (0, -0.5) {$S_{M^\nu}$};
	\end{pgfonlayer}
	\begin{pgfonlayer}{edgelayer}
		\draw [qWire] (6) to (5.center);
	\end{pgfonlayer}
\end{tikzpicture}
}\right) \mapsto %
\InputIfFileExists{Diagrams/probProjc.tikz}{}{\input{./figures/Diagrams/probProjc.tikz}}, \\
  &{\rm and }\nonumber \\ \label{unitfcgptex}
  u_{M^\nu} &= %
\begin{tikzpicture}
	\begin{pgfonlayer}{nodelayer}
		\node [style=none] (5) at (-0.5, 0) {};
		\node [style=none] (6) at (-0.5, -1.5) {};
		\node [style=right label] (14) at (-0.5, -1.25) {$S_{M^\nu}$};
		\node [style=upground] (15) at (-0.5, 0.25) {};
	\end{pgfonlayer}
	\begin{pgfonlayer}{edgelayer}
		\draw [qWire] (6.center) to (5.center);
	\end{pgfonlayer}
\end{tikzpicture}
} := \ \ %
\InputIfFileExists{Diagrams/GPTUnit1c.tikz}{}{\input{./figures/Diagrams/GPTUnit1c.tikz}}.
\end{align}

There is then a close relationship between the accessible GPT fragment $G_{PM}= (\Omega_P, \mathcal{E}_M, {B}_{PM}, u_M)$ characterizing the original scenario 
and the accessible GPT fragment $G_{P^\mu M^\nu}:=(\Omega_{P^\mu}, \mathcal{E}_{M^\nu}, {B}_{P^\mu M^\nu}, u_{M^\nu})$ characterizing any flag-convexification of this scenario.
\begin{prop}\label{prop:FCisCE}
The accessible GPT fragment $G_{P^\mu M^\nu}$ realized by a flag-convexified source-meter pair is cone-equivalent to the accessible GPT fragment $G_{PM}$ realized by the original multisource-multimeter.
\end{prop}
\proof This result is formally proved in Appendix~\ref{provingFCmeansCE}.
\endproof
This result can be seen quite intuitively by viewing the procedure of flag-convexification as a form of classical post-processing, after which the states and effects in the flag-convexified accessible GPT fragment are proportional (as vectors) to those in the original accessible GPT fragment, but scaled down by some strictly positive real numbers. 

We also show (in Appendix~\ref{subfrag}) that the flag-convexified fragment is a \emph{subfragment} of the original fragment, by which we mean that $\Omega_{P^\mu} \subseteq \Omega_P$ and that $\mathcal{E}_{M^\nu} \subseteq \mathcal{E}_{M}$. This fact, however, is not relevant to any of our following results.

\subsection{Example 2: detector inefficiencies}\label{subsec:Inefficiencies}

Consider a multisource-multimeter scenario which is described by some accessible GPT fragment $G$. Contrast this with a new scenario in which the multisource is unchanged, but the multimeter is realized using inefficient detectors. In this case, the accessible GPT fragment describing the scenario with inefficient detectors is cone-equivalent to the original accessible GPT fragment.

To see this formally, we first model detector inefficiency as a particular kind of postprocessing of a multimeter, which maps a multimeter with outcome set $B$ to one with outcome set $B' := B \cup \{*\}$, where the additional outcome is called the \emph{null outcome}.
For each measurement $y \in Y$ in the multimeter, an outcome $b \in B$ for that measurement has probability $\alpha_{by}$ to get flipped to the null outcome, and probability $1-\alpha_{by}$ to remain unchanged. We let the probabilities $1-\alpha_{by}$ be arbitrary weights in $(0,1]$, where we have excluded the case of zero efficiency. Notice that an effect appearing in two distinct measurements $y$ and $y'$ may therefore have different weights. Notice also that the  probability of obtaining the null outcome -- the detectors not firing -- can depend on the state of the system.
In this sense, we have a maximally general model of detector inefficiencies. 

More explicitly, consider a multimeter $M$ with setting $Y$ and outcome $B$. An inefficient version of this detector can always be described as a classical postprocessing  of $M$, denoted $E^\alpha$,
which results in a new multimeter $M^\alpha$,
\beq \label{eq:IneffPP}
\InputIfFileExists{Diagrams/inefficientDetector1.tikz}{}{\input{./figures/Diagrams/inefficientDetector1.tikz}} \ \:= \ \ %
\InputIfFileExists{Diagrams/inefficientDetector.tikz}{}{\input{./figures/Diagrams/inefficientDetector.tikz}}
\eeq
with setting $Y$ and outcome set $B'=B\cup \{*\}$,  where $*$ is the null outcome.
In particular, the postprocessing $E^\alpha$ is defined by
\beq \label{Ealpha}
\InputIfFileExists{Diagrams/inefficientDetector2.tikz}{}{\input{./figures/Diagrams/inefficientDetector2.tikz}}  =  \sum_{\begin{minipage}{.8cm}\scriptsize $b \in B$,\\ $y \in Y$ \end{minipage}} %
\InputIfFileExists{Diagrams/inefficientDetector3.tikz}{}{\input{./figures/Diagrams/inefficientDetector3.tikz}} \left(
\raisebox{10mm}
{$\! \alpha_{by} %
\begin{tikzpicture}
	\begin{pgfonlayer}{nodelayer}
		\node [style=point] (0) at (0, -0.25) {$*$};
		\node [style=none] (1) at (0, 0.75) {};
		\node [style=right label] (2) at (0, 0.5) {$B'$};
	\end{pgfonlayer}
	\begin{pgfonlayer}{edgelayer}
		\draw [cWire] (1.center) to (0);
	\end{pgfonlayer}
\end{tikzpicture}
}\!\!\!+\!(1-\alpha_{by}) %
\begin{tikzpicture}
	\begin{pgfonlayer}{nodelayer}
		\node [style=point] (0) at (0, -0.25) {$b$};
		\node [style=none] (1) at (0, 0.75) {};
		\node [style=right label] (2) at (0, 0.5) {$B'$};
	\end{pgfonlayer}
	\begin{pgfonlayer}{edgelayer}
		\draw [cWire] (1.center) to (0);
	\end{pgfonlayer}
\end{tikzpicture}
}\!$}\right)\!,
\eeq
for some $\alpha_{by} \in [0,1)$.

It follows that:
\begin{prop}
Consider a multisource-multimeter scenario $(P,M)$ and the accessible GPT fragment $G$ associated with it. The accessible GPT fragment $G^\alpha$ which arises if the multimeter is implemented in an inefficient manner from the scenario $(P,M^\alpha)$ is  cone-equivalent to the original accessible GPT fragment.
\end{prop}
 This can be formally established by a proof analogous to the proof in Appendix~\ref{provingFCmeansCE}. 
Intuitively, the result follows from the fact that for every effect in the original  measurement, the corresponding inefficient  version of the measurement contains an effect that is proportional to it, but rescaled by a value in $(0,1]$.

It is worth explaining the distinction between  processings of a measurement that add inefficiency and those that add noise.
 A standard model for 
 adding noise to a measurement is to imagine that with some probability, the outcome that is returned is the one obtained in the original measurement, and otherwise the outcome is sampled uniformly at random.
 More generally, this probability can depend on the outcome as well as on the measurement. 
Formally, this can be described as a postprocessing of a given multimeter as in Eq.~\eqref{eq:IneffPP}, but where the postprocessing is given not by Eq.~\eqref{Ealpha}, but rather by
\beq
\InputIfFileExists{Diagrams/noisyDetector2.tikz}{}{\input{./figures/Diagrams/noisyDetector2.tikz}}  =  \sum_{\begin{minipage}{.8cm}\scriptsize $b \in B$,\\ $y \in Y$ \end{minipage}} %
\InputIfFileExists{Diagrams/inefficientDetector3.tikz}{}{\input{./figures/Diagrams/inefficientDetector3.tikz}} \left(
\raisebox{10mm}
{$\! \beta_{by} %
\begin{tikzpicture}
	\begin{pgfonlayer}{nodelayer}
		\node [style=point] (0) at (0, -0.25) {$u$};
		\node [style=none] (1) at (0, 0.75) {};
		\node [style=right label] (2) at (0, 0.5) {$B$};
	\end{pgfonlayer}
	\begin{pgfonlayer}{edgelayer}
		\draw [cWire] (1.center) to (0);
	\end{pgfonlayer}
\end{tikzpicture}
}\!\!+\!(1-\beta_{by}) %
\begin{tikzpicture}
	\begin{pgfonlayer}{nodelayer}
		\node [style=point] (0) at (0, -0.25) {$b$};
		\node [style=none] (1) at (0, 0.75) {};
		\node [style=right label] (2) at (0, 0.5) {$B$};
	\end{pgfonlayer}
	\begin{pgfonlayer}{edgelayer}
		\draw [cWire] (1.center) to (0);
	\end{pgfonlayer}
\end{tikzpicture}
}\!$}\right)\!,
\eeq
for some $\beta_{by} \in (0,1]$, where $u$ is the uniform distribution over the outcomes $B$. 
 
In an \emph{inefficient} implementation of a measurement $M$, the GPT effects for each non-null outcome are proportional to the GPT effects for the corresponding outcome of $M$. In a \emph{noisy} implementation of a measurement $M$, this proportionality does not hold for one or more effects. 

We depict a given effect space together with an inefficient and a noisy implementation thereof, to illustrate the difference between the two.
\begin{figure}[htb!]
\centering
\beq
\InputIfFileExists{Diagrams/Rotate1.tikz}{}{\input{./figures/Diagrams/Rotate1.tikz}}
\eeq
\caption{ (a) The effect space arising if the measurement whose effect space is shown in (b) is implemented inefficiently. (b) The effect space corresponding to a given  measurement. (c) The effect space arising if the  measurement from (b) is implemented noisily. The cones defined by each effect space are  are the regions whose boundaries are denoted by black arrows. Clearly, the effect spaces in (a) and (b) define the same cone, which is distinct from the one defined by the effect space in (c).
}
 \label{ngpts}
\end{figure}

\subsection{Classicality of cone-equivalent GPTs}

We now turn to the main technical result which underlies the conceptual results of our paper.
\begin{theorem}\label{maintechnicalres}
 If the accessible GPT fragments $G$ and $G'$ are cone-equivalent, then
$G'$ is simplex-embeddable if and only if $G$ is simplex-embeddable. 
\end{theorem}
\proof See Appendix~\ref{app:maintechnicalres}.
\endproof

Hence, given any two cone-equivalent accessible GPT fragments, either both are classically explainable, or neither is.
Since we have shown that cone-equivalent accessible GPT fragments arise in a number of natural experimental scenarios, this result is a useful tool for assessing nonclassicality in these scenarios.

 Note also that Theorem~\ref{maintechnicalres} holds even with the more permissive notion of cone equivalence discussed in Remark~\ref{rem:EquivalentGPT}.\footnote{ 
Furthermore, the proof does not actually use the full geometric structure of accessible GPT fragments but only that of a `GPT fragment', see Ref.~\cite{selby2022open}
 } This follows directly from Theorem~\ref{maintechnicalres} together with the fact that if a GPT is simplex-embeddable, then so too is any equivalent GPT.
  Furthermore, in Appendix~\ref{strongerthm} we prove a slightly stronger version of this theorem using a slightly weaker notion of \emph{simplicial-cone embedding}.

\section{Applications to witnessing nonclassicality}\label{sec:Results}

The following is a simple but powerful corollary of the results we have proven so far. 
\begin{prop}\label{mainprop}
Quantum correlations that are not classically-explainable can arise in prepare-measure scenarios without measurement settings or source settings.
\end{prop}
\begin{proof}
Consider a scenario with a single multisource and a single multimeter, where the accessible GPT fragment is not classically explainable  by virtue of not being simplex-embeddable. From this, define a flag-convexified scenario by flag-convexifying both the multisource and the multimeter. By Proposition~\ref{prop:FCisCE}, the resulting accessible GPT fragment is cone-equivalent to the original accessible GPT fragment.  Theorem~\ref{maintechnicalres} then states that the flag-convexified accessible GPT fragment is not simplex-embeddable.
But the flag-convexified multisource and flag-convexified multimeter have no input settings, so it follows that GPT scenarios can fail to admit of a classical explanation even in scenarios where both the measurements and the preparations have no external inputs. 
\end{proof}

Next, we recall that an operational prepare-measure  scenario admits of a generalized-noncontextual ontological model if and only if the associated accessible GPT fragment is simplex-embeddable~\cite{simplex_forthcoming}. It follows that one can find operational prepare-measure scenarios which have no measurement settings or source settings,  and yet which do not admit of any noncontextual model. One can generate explicit examples by taking any proof of contextuality involving a multisource and a multimeter (or equivalently, a set of preparations and a set of measurements), and flag-convexifying these\footnote{One can generate an even broader set of examples by considering sources and meters which are  merely \emph{operationally equivalent} to the flag-convexified multisource and multimeter. }.

We now state two immediate but conceptually nontrivial corollaries of Proposition~\ref{mainprop}.

Since measurement incompatibility is only possible if one actually \emph{has} multiple measurements, our proofs of contextuality without measurement settings implies the following.
\begin{corollary}
Proofs of the failure of generalized noncontextuality do not require any incompatible measurements.
\end{corollary}

Additionally, since one's choice of measurements and preparations can only be correlated with hidden variables if one actually \emph{has} multiple measurements and preparations, our proofs of contextuality without measurement settings and without preparation settings implies the following. 
\begin{corollary}
Proofs of the failure of  generalized  noncontextuality need not have freedom of choice as an assumption.
\end{corollary}

Finally, consider a scenario wherein one is aiming to implement  some multisource-multimeter pair
 (e.g., to target particular operational equivalences and corresponding noncontextuality inequalities~\cite{mazurek2016experimental,PhysRevLett.115.110403,schmidallinequalities}), but wherein one's measurements are in fact inefficient, in the sense defined in Section~\ref{subsec:Inefficiencies}. As argued in Section~\ref{subsec:Inefficiencies}, 
 the accessible GPT fragment associated to the targeted multimeter-multisource pair and the one associated to the inefficient counterpart thereof 
  are cone-equivalent. 
Assuming that the accessible GPT fragment  of the targeted multisource-multimeter pair is not simplex-embeddable,   
Theorem~\ref{maintechnicalres} then guarantees that the  accessible GPT fragment of its inefficient counterpart  is not simplex-embeddable either, no matter how inefficient one's detectors.

It follows that:
\begin{corollary}
Proofs of the failure of generalized  noncontextuality are possible for any degree of inefficiency of detectors; therefore,  there is no detector loophole for experimental tests of generalized  noncontextuality.
\end{corollary}

If one's detectors are too \emph{noisy} (as opposed to \emph{inefficient}, where the distinction between the two is described in Sec.~\ref{subsec:Inefficiencies}),
 then there is no possibility of a failure of simplex-embeddability, and hence no possibility of a proof of the failure of generalized noncontextuality. This has been discussed at length in Refs.~\cite{PhysRevLett.115.110403,marvian2020inaccessible}.
(It is moreover clear \textit{why} our proof technique does not apply in the case of overly noisy detectors --- the accessible GPT fragment associated with  a multimeter and the one associated to a noisy version thereof 
 are not cone-equivalent, and so Theorem~\ref{maintechnicalres} cannot be applied.)

 \section{Theory-agnostic tomography using accessible GPT fragments}\label{thryagtom}

The framework of GPTs allows one to analyze experimental data without presuming the correctness of quantum theory.  The notion of GPT tomography for states and effects in an experiment was first used in Ref.~\cite{mazurek2016experimental}, consolidated in Ref.~\cite{mazurek2021experimentally}, and further developed in Ref.~\cite{grabowecky2021experimentally}.

In Ref.~\cite{mazurek2016experimental}, it was noted that if one has demonstrated that certain normalized states exist within the GPT governing one's experiment, then all convex mixture of these are known to exist.  
Similarly, if one has demonstrated that certain effects exist within the GPT governing one's experiment, then one has also demonstrated the existence of every effect in the convex hull of (i) this set of effects, (ii) the set of their complements, and (iii) the zero effect and the unit effect. 
Consequently, when looking for a set of states and effects that demonstrate the nonclassicality of what was realized experimentally, e.g., by violating a noncontextuality inequality, one need not use the states and effects that were actually realized in the experiment.  One can instead use \emph{secondary} states and effects 
 whose existence can be inferred from those that were actually realized.
 The accessible GPT fragment describing a given experiment includes all of these secondary states and effects, and thus characterizes all of the states and effects that are known to exist in the GPT by virtue of those that were in fact realized in one's experiment.

 Although previous experimental works considered preparation procedures that yielded normalized states, in certain experiments it is natural to consider states that are never realized deterministically in the experiment, but only with a certain probability, flagged by some outcome variable. In such cases, states are naturally understood as arising from a probabilistic source.  A standard example of this is when the preparations of a system are achieved by `steering'---i.e., when the preparations of one's system are achieved by preparing that system jointly with another and then implementing a measurement on the latter. 
 
If one wishes to analyze the data from an experiment that involves one or more probabilistic sources, then the sources can be modelled by the state space of an accessible GPT fragment (see again Fig.~\ref{fig:tofragment}). Insofar as accessible GPT fragments allow for the possibility of modelling probabilistic sources, they supplement the techniques described in Refs.~\cite{mazurek2016experimental,mazurek2017experimentally,grabowecky2021experimentally}.  Indeed, only if one allows the sort of state space depicted in Fig.~\ref{fig:tofragment} does one have the parametric freedom to describe \emph{not only} the states appearing in a given source \emph{but also} the probabilities with which they appear.

Given that one can conceptualize an arbitrary set of sources and set of measurements as a single multisource and a single multimeter, it follows that the most general prepare-measure experiment constitutes a multisource-multimeter pair, and the latter is precisely what defines an accessible GPT fragment.

It should be noted that prior experiments performing theory-agnostic tomography had detector inefficiencies, but did not aim to quantify these, and thus opted merely to implement GPT tomography on the effects that were realized by \emph{post-selecting} on obtaining a non-null outcome.  If, however, an experimentalist \emph{is} interested in quantifying the probability of obtaining a null outcome, then they simply incorporate a null outcome in the set of GPT effects to which they are fitting. 
 If we are in the special case in which the null outcome contains no information about the system (i.e., when a detector
has the same probability of failing to fire regardless of the
state impinging upon it), then the set of non-null effects form a resolution of a subnormalized version of the unit effect.

Subnormalized versions of the unit effect are analogous to the subnormalized states discussed above, and being interested in the probability of the null outcome is analogous to being interested in the probabilities with which different states appear in a source. In both cases, by refraining from implementing any post-selection in the analysis of the data, one has the possibility of fitting to these probabilities.   

Incorporating parameters for these probabilities can, in principle, change the values of the parameters defining the states and the effects which achieve the best overall fit to the data.  That is, the best-fit states and effects one obtains using the non-post-selected data can in principle be slightly different from those one would obtain using post-selected data.  The non-post-selected fits are to be preferred insofar as the model allows a parametric freedom that is closer to the true freedom that is present in the experimental set-up. As such, we expect the theoretical tools introduced herein to be useful for GPT tomography.

\section{Conclusion and outlook}

We have in this work addressed the question of when particular accessible GPT fragments admit of a \emph{classical} explanation---an embedding into a strictly classical (i.e., simplicial) GPT of some dimension. One could alternatively study the question of which particular accessible GPT fragments admit of a \emph{quantum} explanation---an embedding into a quantum GPT of some dimension. In the special case of unrestricted GPTs, this question was comprehensively answered in Ref.~\cite{mueller2021testing}, which showed that unrestricted GPTs that admit of a quantum explanation are necessarily described by Euclidean Jordan Algebras.
The more general cases (i.e., GPTs which do not satisfy the no-restriction hypothesis, and accessible GPT fragments which are not themselves GPTs) remain to be studied.
This question can be simplified using the fact that quantum explainability will depend only on the cone equivalence class rather than on the specific accessible GPT fragment being considered. This follows from our proof of Theorem~\ref{maintechnicalres}, but where one replaces the simplicial cone with the quantum cone of positive semidefinite operators.

We also expect the notion of accessible GPT
fragments to be useful in the development of a
resource-theoretic approach to quantifying failures of
generalized noncontextuality in prepare-measure scenarios.  As we noted in the introduction, questions about generalized noncontextuality in an operational theory map to questions about simplex embeddability for the corresponding GPT.  In particular, a multisource-multimeter pair in the operational theory fails to admit of a generalized-noncontextual model if and only if the accessible GPT fragment associated to it fails to be simplex embeddable.  Given Theorem~\ref{maintechnicalres}, the distinction between an accessible GPT fragment that is simplex-embeddable, and therefore free within the resource theory, and one that is not simplex-embeddable, hence nonfree, depends only
on the cone-equivalence class of the fragment.
We address this problem in forthcoming work by showing that this perspective reduces testing for freeness to a linear program.

\tocless\section{Acknowledgements}

We thank Victoria J.~Wright for useful discussions on the connection between our work and Ref.~\cite{wright2020general}.  We also thank Ludovico Lami for useful comments on our manuscript, and for pointing out the connection between cone-equivalence and projective-equivalence. 
This research was supported by Perimeter Institute for Theoretical Physics. Research at Perimeter Institute is supported in part by the Government of Canada through the Department of Innovation, Science and Economic Development and by the Province of Ontario through the Ministry of Colleges and Universities. 
DS, and ABS acknowledge support by the Foundation for Polish Science (IRAP project, ICTQT, contract no.2018/MAB/5, co-financed by EU within Smart Growth Operational Programme). 	JHS was supported by the
National Science Centre, Poland (Opus project, Categorical
Foundations of the Non-Classicality of Nature, project
no. 2021/41/B/ST2/03149). RK is supported by the Charg\'e de recherches fellowship of the Fonds de la Recherche Scientifique - FNRS (F.R.S.-FNRS), Belgium.
All of the diagrams within this manuscript were prepared using TikZit.

\bibliographystyle{apsrev4-2-wolfe}
\setlength{\bibsep}{3pt plus 3pt minus 2pt}
\nocite{apsrev42Control}
\bibliography{bib}

\appendix
\renewcommand{\addcontentsline}[2][]{\nocontentsline#1{#2}}

\section[Proof that every effect couples to the unit effect]{Proof of Eq.~\eqref{eq:MeasCompletion}}\label{proofcomplem}

We now prove Eq.~\eqref{eq:MeasCompletion}, namely that in an accessible GPT fragment (just as in a standard GPT), it holds that for any effect $e \in \mathcal{E}_M$, there exists another effect $e^\perp \in \mathcal{E}_M$ such that 
\begin{equation*}
e+e^\perp = u_M.
\end{equation*}

Consider some substochastic comb $d_e$ that realizes effect $e$ as
\beq
\begin{tikzpicture}
	\begin{pgfonlayer}{nodelayer}
		\node [style=none] (5) at (0, -0.75) {};
		\node [style=copoint] (6) at (0, 0.5) {$e$};
		\node [style=right label] (14) at (0, -0.5) {$S_M$};
	\end{pgfonlayer}
	\begin{pgfonlayer}{edgelayer}
		\draw [qWire] (6) to (5.center);
	\end{pgfonlayer}
\end{tikzpicture}
}\ \ =\ \ %
\InputIfFileExists{Diagrams/effectRealisation.tikz}{}{\input{./figures/Diagrams/effectRealisation.tikz}} \,,
\eeq
with
\beq\label{eq:dec}
\InputIfFileExists{Diagrams/effectReal1.tikz}{}{\input{./figures/Diagrams/effectReal1.tikz}} \ \ = \ \ %
\InputIfFileExists{Diagrams/clascombel1.tikz}{}{\input{./figures/Diagrams/clascombel1.tikz}}\,.
\eeq
Now, for any subnormalized probability distribution $q$ on the classical variables $Y,Z$, there exists another subnormalized probability distribution $q'$ on  $Y,Z$ such that the distribution $\tilde{q}:=q+q'$ is normalized. Similarly, for every response function $r$ on the classical variables $B,Z$, there exists a response function $r'$ on $B,Z$ such that $\tilde{r} := r+r'=\mathbf{1}_{B\times Z}$, the unit effect on $B\times Z$. Therefore, on the one hand,
\beq
\InputIfFileExists{Diagrams/clascombel.tikz}{}{\input{./figures/Diagrams/clascombel.tikz}} \ \ = \ \ %
\InputIfFileExists{Diagrams/clascombel5.tikz}{}{\input{./figures/Diagrams/clascombel5.tikz}} \ \ =\ \ %
\InputIfFileExists{Diagrams/effectReal3.tikz}{}{\input{./figures/Diagrams/effectReal3.tikz}}
\eeq
where $p$ is the marginal of $\tilde{q}$ on $Y$ and hence is necessarily normalized, and, on the other hand,
\beq
\InputIfFileExists{Diagrams/clascombel.tikz}{}{\input{./figures/Diagrams/clascombel.tikz}} \ \ = \ \ %
\InputIfFileExists{Diagrams/clascombel1.tikz}{}{\input{./figures/Diagrams/clascombel1.tikz}}  + \left( %
\InputIfFileExists{Diagrams/clascombel2.tikz}{}{\input{./figures/Diagrams/clascombel2.tikz}}  +  %
\InputIfFileExists{Diagrams/clascombel3.tikz}{}{\input{./figures/Diagrams/clascombel3.tikz}}  +  %
\InputIfFileExists{Diagrams/clascombel4.tikz}{}{\input{./figures/Diagrams/clascombel4.tikz}}\right) \,. 
\eeq
Equating these two expressions and identifying the bracketed part of the second as $d_e^\perp$ we have that:
\beq\label{eq:undec}
\InputIfFileExists{Diagrams/effectReal3.tikz}{}{\input{./figures/Diagrams/effectReal3.tikz}} \ \ = \ \ %
\InputIfFileExists{Diagrams/clascombel1.tikz}{}{\input{./figures/Diagrams/clascombel1.tikz}}  +  %
\InputIfFileExists{Diagrams/effectReal2.tikz}{}{\input{./figures/Diagrams/effectReal2.tikz}}\,.
\eeq
The substochastic comb $d_e^\perp$ then realizes the effect $e^\perp$, which we can easily check satisfies $e+e^\perp = u_M$ via:
\begin{align}
\begin{tikzpicture}
	\begin{pgfonlayer}{nodelayer}
		\node [style=none] (5) at (0, -0.75) {};
		\node [style=copoint] (6) at (0, 0.5) {$e$};
		\node [style=right label] (14) at (0, -0.5) {$S_M$};
	\end{pgfonlayer}
	\begin{pgfonlayer}{edgelayer}
		\draw [qWire] (6) to (5.center);
	\end{pgfonlayer}
\end{tikzpicture}
} +\ %
\begin{tikzpicture}
	\begin{pgfonlayer}{nodelayer}
		\node [style=none] (5) at (0, -0.75) {};
		\node [style=copoint] (6) at (0, 0.5) {$e^\perp$};
		\node [style=right label] (14) at (0, -0.5) {$S_M$};
	\end{pgfonlayer}
	\begin{pgfonlayer}{edgelayer}
		\draw [qWire] (6) to (5.center);
	\end{pgfonlayer}
\end{tikzpicture}
}\  = \ \ &%
\InputIfFileExists{Diagrams/effectRealisation.tikz}{}{\input{./figures/Diagrams/effectRealisation.tikz}}\ + %
\InputIfFileExists{Diagrams/effectPerpRealisation.tikz}{}{\input{./figures/Diagrams/effectPerpRealisation.tikz}}\\
\overset{\text{Eqs.\eqref{eq:dec}\&\eqref{eq:undec}}}{=}\ \  &%
\InputIfFileExists{Diagrams/GPTUnit.tikz}{}{\input{./figures/Diagrams/GPTUnit.tikz}}\\
=\ \   &%
\begin{tikzpicture}
	\begin{pgfonlayer}{nodelayer}
		\node [style=none] (5) at (-0.5, 0) {};
		\node [style=none] (6) at (-0.5, -1.5) {};
		\node [style=right label] (14) at (-0.5, -1.25) {$S_M$};
		\node [style=upground] (15) at (-0.5, 0.25) {};
	\end{pgfonlayer}
	\begin{pgfonlayer}{edgelayer}
		\draw [qWire] (6.center) to (5.center);
	\end{pgfonlayer}
\end{tikzpicture}
}.  
\end{align}

\section{Proof of Theorem~\ref{maintechnicalres}}\label{app:maintechnicalres}

Recall that a pair of accessible GPT fragments with the same probability rule and unit, $G=(\Omega,\mathcal{E},{B},u)$ and $G'=(\Omega',\mathcal{E}',{B},u)$, are cone-equivalent iff 
\begin{align}
    \mathsf{Cone}[\Omega'] &= \mathsf{Cone}[\Omega]\\
    \mathsf{Cone}[\mathcal{E}'] &= \mathsf{Cone}[\mathcal{E}].
\end{align}

We now prove our main technical result, repeated here:
\addtocounter{thmrepeat}{\getrefnumber{maintechnicalres}}
\addtocounter{thmrepeat}{-1}
\begin{thmrepeat}
If the accessible GPT fragments $G$ and $G'$ are cone-equivalent, then
$G'$ is simplex-embeddable if and only if $G$ is simplex-embeddable. 
\end{thmrepeat}
\proof

As cone equivalence is a symmetric notion, we only need to prove one direction. We now consider the case where $G'$ is simplex-embeddable, and we show that this implies that $G$ is simplex-embeddable. 

By assumption, then, there exists some integer $n$ and linear maps $\iota:\mathsf{Span}[\Omega'] \to \mathds{R}^n$ and $\kappa:\mathsf{Span}[\mathcal{E}'] \to {\mathds{R}^n}^*$ such that for all $s'\in\Omega'$ and $e'\in \mathcal{E}'$, one has
\begin{align}
\iota(s') &\in \Delta_{sub}^n \label{embcond1},\\
\kappa(e') &\in {\Delta_{sub}^n}^* \label{embcond2},\\
\kappa(e')[\iota(s')]  &= {B}(e',s')\label{embcond3},\\
\kappa(u)&=\mathbf{1}_n \label{embcond4}.
\end{align}
We now show that the embedding maps $\iota$ and $\kappa$ also constitute a simplex embedding for $G$,
i.e., that for all $s\in\Omega$ and $e\in \mathcal{E}$, we have
\begin{align}
\iota(s) &\in \Delta_{sub}^n \label{embcond1p},\\
\kappa(e) &\in {\Delta_{sub}^n}^* \label{embcond2p},\\
\kappa(e)[\iota(s)]  &= {B}(e,s) \label{embcond3p},\\
\kappa(u)&=\mathbf{1}_n. \label{embcond4p}
\end{align}

First, we note that the unit effects for $G$ and $G'$ are (by assumption) the same, and so Eq.~\eqref{embcond4} and Eq.~\eqref{embcond4p} are exactly the same condition.

Next, we show that Eq.~\eqref{embcond3} implies Eq.~\eqref{embcond3p}.
Noting that $\mathsf{Span}[X]=\mathsf{Span}[\mathsf{Cone}[X]]$ and recalling that (by cone equivalence) $\mathsf{Cone}[\Omega'] = \mathsf{Cone}[\Omega]$ and 
    $\mathsf{Cone}[\mathcal{E}'] = \mathsf{Cone}[\mathcal{E}]$, it follows that
\begin{align}
\mathsf{Span}[\Omega']&=\mathsf{Span}[\Omega]\\
\mathsf{Span}[\mathcal{E}']&=\mathsf{Span}[\mathcal{E}].
\end{align}
It follows that every effect $e\in \mathcal{E}$ can be decomposed as a linear combination of effects in $\mathcal{E}'$, i.e., 
\beq 
e = \sum_i \alpha_i e'_i \text{ \ \ where } \alpha_i \in \mathds{R}  \text{ and } e'_i \in \mathcal{E}',
\eeq
and similarly, that every state $s\in\Omega$ can be decomposed as a linear combination of states in $\Omega'$, i.e.,
\beq s=\sum_j \beta_j s'_j  \text{ \ \ where } \beta_j \in \mathds{R} \text{ and } s'_i \in \Omega'.
\eeq
It follows that:
\begin{align}
{B}(e,s) &= \sum_{ij} \alpha_i \beta_j\ {B}(e'_i,s'_j) \\
&= \sum_{ij} \alpha_i \beta_j\ \kappa(e'_i)[\iota(s'_j)] \\
&= \kappa(e)[\iota(s)], \label{randomeqnlabel2}
\end{align}
where in the first equality we use the decompositions of $e$ and $s$ together with bilinearity of ${B}$, in the second we use Eq.~\eqref{embcond3}, and in the third we use bilinearity of composition and linearity of $\kappa$ and $\iota$.

Next, we show that Eq.~\eqref{embcond1} implies Eq.~\eqref{embcond1p}, so that
$\iota(s)\in \Delta_{sub}^n$ for all states $s\in \Omega$.
Note that $\Delta_{sub}^n$ is equal to the intersection of $\mathsf{Cone}[\Delta_{sub}^n]$ with the space
\beq 
\{v\in \mathds{R}^n : \mathbf{1}_n[v] \leq 1\}.
\eeq
of subnormalized vectors in the vector space $\mathds{R}^n$. 
Now, Eq.~\eqref{embcond1} states that 
\beq \iota(\Omega') \subseteq \Delta_{sub}^n,
\eeq 
 and so
\beq
\mathsf{Cone}[\iota(\Omega')] \subseteq \mathsf{Cone}[\Delta_{sub}^n],
\eeq
which by linearity of $\iota$ implies that
\beq \iota(\mathsf{Cone}[\Omega']) \subseteq \mathsf{Cone}[\Delta_{sub}^n],
\eeq 
and as 
$\mathsf{Cone}[\Omega'] = \mathsf{Cone}[\Omega]$,
we therefore have that 
\beq
\iota(\mathsf{Cone}[\Omega]) \subseteq \mathsf{Cone}[\Delta_{sub}^n].
\eeq 
In particular, this means that $\iota(s) \in \mathsf{Cone}[\Delta_{sub}^n]$ for every $s \in \mathsf{Cone}[\Omega]$. All that remains is to show that $\iota(s)$ is a subnormalized vector, i.e, that $\mathbf{1}_n[\iota(s)] \leq 1$. Eq.~\eqref{embcond4} tells us that $\mathbf{1}_n = \kappa(u)$; using this together with the fact (derived above in Eq.~\eqref{randomeqnlabel2}) that $\kappa(e)[\iota(s)] = {B}(e,s)$), we find that
\begin{align}
\mathbf{1}_n[\iota(s)] &= \kappa(u)[\iota(s)]  \\
&= {B}(u,s) \\
&\leq 1,
\end{align}
where the last line follows from the definition of $u$ and the fact that $s\in \Omega$.
So $\iota(s)$ is indeed subnormalized, and Eq.~\eqref{embcond1} follows.

Finally, we show that Eq.~\eqref{embcond2} implies Eq.~\eqref{embcond2p}.
Eq.~\eqref{embcond2} states that 
\beq \kappa(\mathcal{E}') \subseteq {\Delta_{sub}^n}^*,
\eeq 
which by linearity of $\kappa$ implies that
\beq \kappa(\mathsf{Cone}[\mathcal{E}']) \subseteq \mathsf{Cone}[{\Delta_{sub}^n}^*],
\eeq 
and as 
$\mathsf{Cone}[\mathcal{E}'] = \mathsf{Cone}[\mathcal{E}]$
we have that 
\beq\label{kappaembed}
\kappa(\mathsf{Cone}[\mathcal{E}]) \subseteq \mathsf{Cone}[{\Delta_{sub}^n}^*].
\eeq 
In particular, this means that $\kappa(e) \in \mathsf{Cone}[{\Delta_{sub}^n}^*]$ for every $e \in \mathsf{Cone}[\mathcal{E}]$.
We can then note that ${\Delta_{sub}^n}^*$ is equal to the set of vectors which are in  $\mathsf{Cone}[{\Delta_{sub}^n}^*]$ and which are below $\mathbf{1}_n$ in the partial order defined by $\mathsf{Cone}[{\Delta_{sub}^n}^*]$ via
\beq 
v\leq v'\quad  \iff\quad  \exists c \in \mathsf{Cone}[{\Delta_{sub}^n}^*] \text{ s.t. } v+c = v'.
\eeq 
As we already know that $\kappa(e) \in \mathsf{Cone}[{\Delta_{sub}^n}^*]$  for $e \in \mathcal{E}$, all that we must show is that $\kappa(e) \leq \mathbf{1}_n$ with respect to this partial order. That is, we need to show that there exists $c_e\in \mathsf{Cone}[{\Delta_{sub}^n}^*]$ such that 
\beq\label{identifyingce}
\kappa(e) + c_e = \mathbf{1}_n.
\eeq
Recall from Eq.~\eqref{eq:MeasCompletion} that for any effect $e \in \mathcal{E}$ in an accessible GPT fragment, there exists some effect $e^\perp \in \mathcal{E}$  such that 
\beq 
e+e^\perp = u.
\eeq 
Applying $\kappa$ to this equation, and using linearity of $\kappa$, we obtain 
\beq 
\kappa(e) + \kappa(e^\perp) = \kappa(u).
\eeq
By Eq.~\eqref{embcond4p}, $\kappa(u) = \mathbf{1}_n$. Using the fact that $e^\perp \in \mathcal{E}$ and hence (by Eq.~\eqref{kappaembed}) that $\kappa(e^\perp) \in \mathsf{Cone}[{\Delta_{sub}^n}^*]$, we can then see that the choice $c_e := \kappa(e^\perp)$ satisfies Eq.~\eqref{identifyingce}, and so $\kappa(e) \leq \mathbf{1}_n$. This shows that $\kappa(e) \in {\Delta_{sub}^n}^*$, demonstrating Eq.~\eqref{embcond2p}.

Putting all this together, what we have shown is that the simplex embedding maps $\iota$ and $\kappa$ for $G'$ are necessarily also simplex embedding maps for $G$, and so, if $G'$ is simplex-embeddable, then so too is $G$. 
\endproof

\section{Flag-convexification yields a cone-equivalent accessible GPT fragment} \label{provingFCmeansCE}

Recall the accessible GPT fragment $G_{P^\mu M^\nu}=(\Omega_{P^\mu}, \mathcal{E}_{M^\nu}, {B}_{P^\mu M^\nu}, u_{M^\nu})$ defined in Eqs.~\eqref{eq:fcGPT}-\eqref{unitfcgptex}, arising from the flag-convexification of some multisource-multimeter pair $(P,M)$.
We will now show that $G_{P^\mu M^\nu}$ is  cone-equivalent to the accessible GPT fragment $G_{PM}= (\Omega_P, \mathcal{E}_M, {B}_{PM}, u_M)$ associated with the original multisource-multimeter pair. 

First, we relate the relevant subspaces:
\begin{lemma}
The vector space spanned by the flag-convexified accessible GPT fragment's states (effects) is equal to the vector space spanned by the original accessible GPT fragment's states (effects). That is, $S_P=S_{P^\mu}$ and  $S_{M}=S_{M^\nu}$.
\end{lemma}
\proof
First note that the span of all of the accessible states is the same as the span of the states accessible by ``atomic'' substochastic combs:
\beq
S_P := \mathsf{Span}\left[\left\{%
\InputIfFileExists{Diagrams/OmegaP3.tikz}{}{\input{./figures/Diagrams/OmegaP3.tikz}} \right\}\right] = \mathsf{Span}\left[\left\{ %
\InputIfFileExists{Diagrams/RayP.tikz}{}{\input{./figures/Diagrams/RayP.tikz}} \right\}\right] ,
\eeq
and similarly in the flag convexified case:
\beq
S_{P^\mu} :=  \mathsf{Span}\left[\left\{%
\InputIfFileExists{Diagrams/fcnuOmegac.tikz}{}{\input{./figures/Diagrams/fcnuOmegac.tikz}} \right\}\right] = \mathsf{Span}\left[\left\{ %
\InputIfFileExists{Diagrams/span1.tikz}{}{\input{./figures/Diagrams/span1.tikz}} \right\}\right].
\eeq
 This is because $c$ and $c'$ can always be decomposed as 
\begin{align}\label{decompceffect}
\InputIfFileExists{Diagrams/classComb.tikz}{}{\input{./figures/Diagrams/classComb.tikz}} &= \sum_{ax} c_{ax} %
\InputIfFileExists{Diagrams/classCombDecomp.tikz}{}{\input{./figures/Diagrams/classCombDecomp.tikz}} \text{ and} \\ \label{decompceffectfc} %
\InputIfFileExists{Diagrams/classCombc.tikz}{}{\input{./figures/Diagrams/classCombc.tikz}} &= \sum_{ax} c'_{ax} %
\InputIfFileExists{Diagrams/classCombDecompc.tikz}{}{\input{./figures/Diagrams/classCombDecompc.tikz}} \text{, respectively,}
\end{align}
for some $c_{ax}, c'_{ax}\in \mathds{R}^+$.
Then, we can write
\beq
\InputIfFileExists{Diagrams/span1.tikz}{}{\input{./figures/Diagrams/span1.tikz}} = %
\InputIfFileExists{Diagrams/fcnuRayP.tikz}{}{\input{./figures/Diagrams/fcnuRayP.tikz}} = %
\InputIfFileExists{Diagrams/RayP.tikz}{}{\input{./figures/Diagrams/RayP.tikz}} %
\begin{tikzpicture}
	\begin{pgfonlayer}{nodelayer}
		\node [style=point] (45) at (0., -0.5) {$\mu$};
		\node [style=none] (47) at (0., 0) {};
		\node [style=copoint] (48) at (0., 0.5) {$x$};
	\end{pgfonlayer}
	\begin{pgfonlayer}{edgelayer}
		\draw [cWire] (48) to (47.center);
		\draw [cWire] (47.center) to (45);
	\end{pgfonlayer}
\end{tikzpicture}
},
\eeq
and hence
\beq
S_{P^\mu} =  \mathsf{Span}\left[\left\{%
\InputIfFileExists{Diagrams/RayP.tikz}{}{\input{./figures/Diagrams/RayP.tikz}} %
\begin{tikzpicture}
	\begin{pgfonlayer}{nodelayer}
		\node [style=point] (45) at (0., -0.5) {$\mu$};
		\node [style=none] (47) at (0., 0) {};
		\node [style=copoint] (48) at (0., 0.5) {$x$};
	\end{pgfonlayer}
	\begin{pgfonlayer}{edgelayer}
		\draw [cWire] (48) to (47.center);
		\draw [cWire] (47.center) to (45);
	\end{pgfonlayer}
\end{tikzpicture}
} \right\}\right].
\eeq
Since $\mu$ has full support, $%
\begin{tikzpicture}
	\begin{pgfonlayer}{nodelayer}
		\node [style=point] (45) at (0., -0.5) {$\mu$};
		\node [style=none] (47) at (0., 0) {};
		\node [style=copoint] (48) at (0., 0.5) {$x$};
	\end{pgfonlayer}
	\begin{pgfonlayer}{edgelayer}
		\draw [cWire] (48) to (47.center);
		\draw [cWire] (47.center) to (45);
	\end{pgfonlayer}
\end{tikzpicture}
} \neq 0$, and so 
\beq
S_{P^\mu} =  \mathsf{Span}\left[\left\{%
\InputIfFileExists{Diagrams/RayP.tikz}{}{\input{./figures/Diagrams/RayP.tikz}} \right\}\right] = S_P.
\eeq
The proof for the measurement side is exactly analogous.
\endproof

This has immediate consequences for the structure of the two GPTs. 
The fact that the subspaces are the same immediately implies that the inclusion maps and projection maps are the same. That is, the fact that $S_P = S_{P^\mu}$ means that
\beq
\InputIfFileExists{Diagrams/inc1.tikz}{}{\input{./figures/Diagrams/inc1.tikz}} =\ %
\InputIfFileExists{Diagrams/inc2.tikz}{}{\input{./figures/Diagrams/inc2.tikz}} \ \text{and} \ \ %
\InputIfFileExists{Diagrams/projn1.tikz}{}{\input{./figures/Diagrams/projn1.tikz}} =\ %
\InputIfFileExists{Diagrams/projn2.tikz}{}{\input{./figures/Diagrams/projn2.tikz}}
\eeq
and the fact that $S_M = S_{M^\nu}$ means that
\beq
\InputIfFileExists{Diagrams/inc3.tikz}{}{\input{./figures/Diagrams/inc3.tikz}} =\ %
\InputIfFileExists{Diagrams/inc4.tikz}{}{\input{./figures/Diagrams/inc4.tikz}} \ \text{and} \ \ %
\InputIfFileExists{Diagrams/projn3.tikz}{}{\input{./figures/Diagrams/projn3.tikz}} =\ %
\InputIfFileExists{Diagrams/pojn4.tikz}{}{\input{./figures/Diagrams/pojn4.tikz}}.
\eeq
 Given that the probability rule and the unit effect for the accessible GPT fragment can be defined in terms of these inclusion and projection maps (as is done in Eqs.~\eqref{eq:GPTProb} and \eqref{defnunit}, respectively) it therefore follows that ${B}_{P^\mu M^\nu} = {B}_{PM}$ and $u_{M^\nu}=u_M$. 

We can now re-express the state space and effect space (Eq.~\eqref{eq:fcGPT}) for the accessible GPT fragment in the flag-convexified scenario by using the projectors from the original scenario:
\begin{align}
\Omega_{P^\mu} &= \left\{%
\begin{tikzpicture}
	\begin{pgfonlayer}{nodelayer}
		\node [style=none] (5) at (0, 0.5) {};
		\node [style=point] (6) at (0, -0.75) {$s$};
		\node [style=right label] (14) at (0, 0.25) {$S_P$};
	\end{pgfonlayer}
	\begin{pgfonlayer}{edgelayer}
		\draw [qWire] (6) to (5.center);
	\end{pgfonlayer}
\end{tikzpicture}
}  = 
%:= \tikzfig{Diagrams/stateProj} \middle| \tikzfig{Diagrams/stateP1} = \tikzfig{Diagrams/fcnuOmegac} 
%
\InputIfFileExists{Diagrams/cone1.tikz}{}{\input{./figures/Diagrams/cone1.tikz}} \right\}, \\
\mathcal{E}_{M^\nu} &=  \left\{%
\begin{tikzpicture}
	\begin{pgfonlayer}{nodelayer}
		\node [style=none] (5) at (0, -0.75) {};
		\node [style=copoint] (6) at (0, 0.5) {$e$};
		\node [style=right label] (14) at (0, -0.5) {$S_M$};
	\end{pgfonlayer}
	\begin{pgfonlayer}{edgelayer}
		\draw [qWire] (6) to (5.center);
	\end{pgfonlayer}
\end{tikzpicture}
} = 
%:= \tikzfig{Diagrams/effectProj} \middle| \tikzfig{Diagrams/effectM1} = \tikzfig{Diagrams/fcnucalEc} 
%
\InputIfFileExists{Diagrams/simpEffectSpaceFC.tikz}{}{\input{./figures/Diagrams/simpEffectSpaceFC.tikz}}\right\}.
\end{align}
We now use this to show that the state and effect cones for these two GPTs are identical. 
\begin{lemma}
The cone generated by the state (effect) space of a flag-convexified accessible GPT fragment is equal to the cone generated by the state (effect) space of the original accessible GPT fragment. That is, $\mathsf{Cone}[\Omega_P] = \mathsf{Cone}[\Omega_{P^\mu}]$ and $\mathsf{Cone}[\mathcal{E}_M] = \mathsf{Cone}[\mathcal{E}_{M^\nu}]$.
\end{lemma}
\proof
This is shown by straightforward calculation:
\begin{align}
\mathsf{Cone}[\Omega_{P^\mu}] &= \mathsf{Cone}\left[\left\{ %
\InputIfFileExists{Diagrams/cone1.tikz}{}{\input{./figures/Diagrams/cone1.tikz}} \right\}_{\! \! c'}\right]  \\
&= \mathsf{Cone}\left[\left\{ %
\InputIfFileExists{Diagrams/cone2.tikz}{}{\input{./figures/Diagrams/cone2.tikz}} \right\}_{\!\!  a,x}\right]  \label{randomeqexpln}\\
&= \mathsf{Cone}\left[\left\{%
\InputIfFileExists{Diagrams/cone3.tikz}{}{\input{./figures/Diagrams/cone3.tikz}} \right\}_{\!\!  a,x}\right]  \\
&= \mathsf{Cone}\left[\left\{%
\InputIfFileExists{Diagrams/cone4.tikz}{}{\input{./figures/Diagrams/cone4.tikz}}%
\begin{tikzpicture}
	\begin{pgfonlayer}{nodelayer}
		\node [style=point] (45) at (0., -0.5) {$\mu$};
		\node [style=none] (47) at (0., 0) {};
		\node [style=copoint] (48) at (0., 0.5) {$x$};
	\end{pgfonlayer}
	\begin{pgfonlayer}{edgelayer}
		\draw [cWire] (48) to (47.center);
		\draw [cWire] (47.center) to (45);
	\end{pgfonlayer}
\end{tikzpicture}
} \right\}_{\!\!  a,x}\right]  \\
&= \mathsf{Cone}\left[\left\{%
\InputIfFileExists{Diagrams/cone4.tikz}{}{\input{./figures/Diagrams/cone4.tikz}} \right\}_{\!\!  a,x}\right] \\
&=\mathsf{Cone}[\Omega_{P}].
\end{align}
Here, Eq.~\eqref{randomeqexpln} follows by substituting a decomposition of $c'$ as shown  in Eq.~\eqref{decompceffectfc}. 
The proof for the measurement side is exactly analogous.
\endproof

This establishes cone equivalence.

\section{Proof that a flag-convexified accessible GPT fragment is a subfragment of the original accessible GPT fragment} \label{subfrag}

We here prove the following:
\begin{lemma}
The state (effect) space of a flag-convexified accessible GPT fragment is contained within the state (effect) space of the original accessible GPT fragment. That is, 
 $\Omega_{P^\mu} \subseteq \Omega_P$ and $\mathcal{E}_{M^\nu} \subseteq \mathcal{E}_M$.
\end{lemma}
\proof
Start by assuming that we have some state $s \in \Omega_{P^\mu}$. The definition of $\Omega_{P^\mu}$ means that
\beq
\begin{tikzpicture}
	\begin{pgfonlayer}{nodelayer}
		\node [style=none] (5) at (0, 0.5) {};
		\node [style=point] (6) at (0, -0.75) {$s$};
		\node [style=right label] (14) at (0, 0.25) {$S_P$};
	\end{pgfonlayer}
	\begin{pgfonlayer}{edgelayer}
		\draw [qWire] (6) to (5.center);
	\end{pgfonlayer}
\end{tikzpicture}
} = %
\InputIfFileExists{Diagrams/fcnuOmegac2.tikz}{}{\input{./figures/Diagrams/fcnuOmegac2.tikz}} = %
\InputIfFileExists{Diagrams/fcnuOmegad.tikz}{}{\input{./figures/Diagrams/fcnuOmegad.tikz}}.
\eeq
Now, we can make the identification
\beq
\InputIfFileExists{Diagrams/cp_to_c2.tikz}{}{\input{./figures/Diagrams/cp_to_c2.tikz}} =: %
\InputIfFileExists{Diagrams/OmegaPc.tikz}{}{\input{./figures/Diagrams/OmegaPc.tikz}} ,
\eeq
which immediately means that $s \in \Omega_P$.

The proof for the measurement side is exactly analogous.
\endproof

\section{Alternative Proof of Theorem~\ref{maintechnicalres}} \label{strongerthm}

An accessible GPT fragment $G=(\Omega,\mathcal{E},B,u)$ is simplex embeddable if and only if there exists some integer $n$ and linear maps $\iota:S_P \to \mathds{R}^n$ and $\kappa:{S_M}^* \to {\mathds{R}^n}^*$ such that for all $s\in\Omega$ and $e\in \mathcal{E}$, one has
\begin{align}
\iota(s) &\in \Delta_{sub}^n \label{embcond11},\\
\kappa(e) &\in {\Delta_{sub}^n}^* \label{embcond21},\\
\kappa(e)[\iota(s)]  &= {B}(e,s)\label{embcond31},\\
\kappa(u)&=\mathbf{1}_n \label{embcond41}.
\end{align}

Now, we can define an apparently weaker notion which refers only to cones. That is, we say that an accessible GPT fragment $G=(\Omega,\mathcal{E},B,u)$ is {\bf simplicial-cone embeddable} if and only if there exists some integer $n$ and linear maps $\iota:\mathsf{Span}[\Omega] \to \mathds{R}^n$ and $\kappa:\mathsf{Span}[\mathcal{E}] \to {\mathds{R}^n}^*$ such that for all $s\in\mathsf{Cone}[\Omega]$ and $e\in \mathsf{Cone}[\mathcal{E}]$, one has
\begin{align}
\iota(s) &\in \mathsf{Cone}[\Delta_{sub}^n] \label{embcond12},\\
\kappa(e) &\in \mathsf{Cone}[{\Delta_{sub}^n}]^* \label{embcond22},\\
\kappa(e)[\iota(s)]  &= {B}(e,s)\label{embcond32}.
\end{align}

As this condition only depends on the triple $(\mathsf{Cone}[\Omega],\mathsf{Cone}[\mathcal{E}],B)$, it is immediately clear that given two cone equivalent accessible GPT fragments, either both will be simplicial-cone embeddable, or neither.

We will now prove the following theorem:
\begin{theorem}
An accessible GPT fragment $G=(\Omega, \mathcal{E}, B, u)$ is simplex embeddable if and only if it is simplicial-cone embeddable.
\end{theorem}
\proof
The only if direction is trivial, as the embedding maps for simplex embedding also satisfy all of the required conditions for simplicial-cone embedding. We will therefore focus on the if direction; that is, the existence of an $n$, $\iota$, and $\kappa$ defining a simplicial-cone embedding implies the existence of an $n'$, $\iota'$, and $\kappa'$ defining a simplex embedding.

 The key step in the proof is to show that although it need not be the case that $\kappa(u)=\mathbf{1}_n$, the $\kappa'$ that we construct does satisfy $\kappa'(u)=\mathbf{1}_n$. 

Note that without loss of generality we can take $\kappa(u)$ to be in the interior of the simplicial cone $\mathsf{Cone}[{\Delta_{sub}^n}]^*$. The reason that this can be done without loss of generality is that if $\kappa(u)$ is not in the interior, then it must lie on a simplicial subcone constituting a face. Then, one can always project down onto the simplicial-subcone for which it is in the interior, and then have a new simplicial-cone embedding into this subcone. From the perspective of ontological models, what this means is that if $\kappa(u)$ does not have full support on the ontic state space $\Lambda$, then we can always restrict attention to the subset of $\Lambda$ for which it does have support, as all of the other response functions in the model necessarily only have support on this subset, and so we can project the epistemic states of the model into this space as well. More formally, note that every other effect $\kappa(e)$ is already in this subcone, as we have that there exists $\kappa(e^\perp)$ in the cone such that $\kappa(e)+\kappa(e^\perp) = \kappa(e+e^\perp) = \kappa(u)$, and hence, when we project $\iota(s)$ into this subcone none of the probabilities of obtaining these effects will change.

We therefore consider the case in which $\kappa(u)$ is in the interior of  $\mathsf{Cone}[{\Delta_{sub}^n}]^*$. As $\mathsf{Cone}[{\Delta_{sub}^n}]^*$ is a homogeneous cone\footnote{A homogeneous cone is one for which the automorphism group is transitive on the interior, i.e., for which given any pair of interior points there is an automorphism mapping between them.}, this means that there exists a cone automorphism\footnote{That is, an invertible linear map which preserves the cone and whose inverse preserves the cone.} $\tau$ such that $\tau^*(\kappa(u)) = \mathbf{1}_n$. We therefore define $\kappa':= \tau^* \circ \kappa$ and $\iota':= {\tau^{-1}} \circ \iota$. These satisfy: 
\begin{align}
\iota'(s) &\in \mathsf{Cone}[\Delta_{sub}^n],\\
\kappa'(e) &\in \mathsf{Cone}[{\Delta_{sub}^n}]^*,\\
\kappa'(e)[\iota'(s)]  &=  \tau^*(\kappa(e))[{\tau^{-1}}(\iota(s))]\\ &=  \kappa(e)[\tau(\tau^{-1}(\iota(s)))]\\ &= {B}(e,s),\\
\kappa'(u)&= \tau^*(\kappa(u))\\ &=\mathbf{1}_n.
\end{align}

 To complete the proof, we need only show that $\iota'(s) \in \Delta_{sub}^n$ and $\kappa'(e) \in {\Delta_{sub}^n}^*$. 

Because we know that $\iota'(s) \in \mathsf{Cone}[\Delta_{sub}^n]$ and we know that
\beq
\mathbf{1}_n[\iota'(s)] = \kappa'(u)[\iota'(s)] = B(u,s) \leq 1,
\eeq
for any $s\in \Omega$, which together mean that $\iota'(s) \in \Delta_{sub}^n$ as we require of a simplex embedding. 

Similarly, we know that for any effect $e$ there exists an effect $e^\perp$ such that $u=e+e^\perp$, and that $\kappa'(e),\kappa'(e^\perp) \in  \mathsf{Cone}[{\Delta_{sub}^n}]^*$, Moreover,
\beq
\kappa'(e)+\kappa'(e^\perp) = \kappa'(u) = \mathbf{1}_n,
\eeq
which means that $\kappa'(e)\leq \mathbf{1}_n$, and, hence, that $\kappa'(e) \in {\Delta_{sub}^n}^*$ as we require of a simplex embedding.

\endproof

What this shows is that whether or not an accessible GPT fragment is simplex embeddable is only actually dependent on the tuple $(\mathsf{Cone}[\Omega],\mathsf{Cone}[\mathcal{E}], B)$, as this is the data necessary to decide if it is simplicial-cone embeddable. 
 Hence, for any two accessible GPT fragments that are cone-equivalent, either (i) they are both simplicial-cone embeddable and both simplex-embeddable, or (ii) neither is simplicial-cone embeddable and neither is simplex-embeddable.

\end{document}